\def\llncs{0}
\def\fullpage{1}
\def\anonymous{0}
\def\draft{0}
\def\submission{0}

\ifnum\submission=1
\def\llncs{1}
\def\draft{0}
\def\anonymous{1}
\def\fullpage{0}
\fi

\ifnum\llncs=1
    \documentclass[envcountsect,envcountsame]{llncs}
    \ifnum\fullpage=1
    \usepackage{fullpage}
    \fi
\else
    \documentclass[letterpaper,hmargin=1.05in,vmargin=1.05in]{article}
    \ifnum\fullpage=1
    \usepackage{fullpage}
    \fi
    \usepackage{microtype}
\fi

\usepackage{authblk}
\usepackage{braket}
\usepackage{xcolor} 
\usepackage{amsmath}
\allowdisplaybreaks[4] 
\usepackage{amssymb}
\usepackage{mathtools}

 \usepackage{amsthm}
 \usepackage{thmtools}
 \usepackage{physics}
\usepackage{enumitem} 
\usepackage[colorlinks=true,linkcolor=magenta,citecolor=blue,hypertexnames=false,pdftex,pdfpagelabels,bookmarks,hyperindex,hyperfigures]{hyperref}
\let\subparagraph\paragraph
\usepackage{titlesec}
\titleformat*{\paragraph}{\normalsize\bfseries}

\usepackage[capitalise,nameinlink]{cleveref}

\ifnum\draft=1
	\newcommand{\minki}[1]{\textcolor{blue}{$\langle\langle$Minki: #1$\rangle\rangle$}}
\else
	\newcommand{\minki}[1]{}
\fi

\ifnum\llncs=0
	\newtheorem{theorem}{Theorem}[section]
	\newtheorem{lemma}[theorem]{Lemma}
	\newtheorem{corollary}[theorem]{Corollary}
	
	\newtheorem{definition}[theorem]{Definition}
	
	\newtheorem{claim}[theorem]{Claim}

	\theoremstyle{remark}

\else
	\spnewtheorem{algorithm}{Algorithm}{\bfseries}{\rmfamily}

	\spnewtheorem{claim}{Claim}{\bfseries}{\itshape}
\fi
\crefname{appendix}{Appendix}{Appendices}
\Crefname{appendix}{Appendix}{Appendices}

\newcommand{\Haar}{\mathcal{H}}

\newcommand{\N}{\mathbb{N}}
\newcommand{\R}{\mathbb{R}}
\newcommand{\C}{\mathbb{C}}

\newcommand{\cH}{\mathcal{H}}

\newcommand{\eps}{\varepsilon}

\newcommand{\EE}{\mathop{\mathbb{E}}}

\renewcommand{\epsilon}{\varepsilon}

\newcommand{\Exp}{\mathbb{E}}

\newcommand{\poly}{\operatorname{poly}}
\newcommand{\negl}{\operatorname{negl}}

\newcommand{\bit}{\{0,1\}}

\ifnum\draft=1
	
\else
	
\fi

\newcommand{\proj}[1]{\ket{#1}\!\bra{#1}}
\renewcommand{\norm}[1]{\left\lVert#1\right\rVert}
\newcommand{\normone}[1]{\left\lVert #1 \right\rVert_1}
\newcommand{\normtwo}[1]{\left\lVert #1 \right\rVert_2}

\DeclarePairedDelimiterX{\inner}[1]{\langle}{\rangle}{\innerp{#1}}

\ExplSyntaxOn

\NewDocumentCommand \innerp { m }
  {
    \int_case:nnF { \clist_count:n {#1} }
      {
        { 1 } { #1, #1 }
        { 2 } {#1}
      }
      {
        \msg_error:nnnx { innerp } { too-many-operands }
          {#1} { \clist_count:n {#1} }
      }
  }
\ExplSyntaxOff

\newcommand{\Prb}{\mathbb{P}}
\newcommand{\normop}[1]{\left\lVert#1\right\rVert_{\infty}}
\newcommand{\normdiamond}[1]{\left\lVert#1\right\rVert_{\diamond}}

\crefname{assumption}{Assumption}{Assumptions}
\Crefname{assumption}{Assumption}{Assumptions}

\newcommand{\cL}{\mathsf{L}}
\newcommand{\Herm}{\operatorname{Herm}}
\newcommand{\Id}{I}

\newcommand{\Var}{\operatorname{Var}}

\newcommand{\Sphere}{\mathbb{S}}

\newcommand{\QPSPACE}{\mathsf{QPSPACE}}

\newcommand{\Gen}{\mathsf{Gen}}

\newcommand{\defeq}{\vcentcolon=}

\newcommand{\simmapsto}[1]{%
  \mathrel{%
    \ooalign{%
      $\xmapsto{#1}$\cr
      \hidewidth\raisebox{-3pt}{$\scriptstyle\approx$}\hidewidth\cr
    }%
  }%
}

\title{Derivatives of Quantum Randomness: Separating Pseudorandom Unitaries from Pseudorandom (Function-like) States}
\hypersetup{pdftitle={Derivatives of Quantum Randomness: Separating Pseudorandom Unitaries from Pseudorandom (Function-like) States}}
\date{}

\ifnum\llncs=1
    \ifnum\anonymous=1
        \author{}
        \institute{}
    \else
        \author{
        Minki Hhan\inst{1} 
        }
        \institute{
        KAIST, Daejeon, Korea
        \\\email{minkihhan@kaist.ac.kr}
        }
    \fi
\else
    \author{Minki Hhan}
    \affil{{\small KAIST, Daejeon, Korea}
    \authorcr{\small minkihhan@kaist.ac.kr}}
    
\fi

\begin{document}

\maketitle

\begin{abstract}
    Quantum computation gives rise to new pseudorandom primitives for states and unitaries, including pseudorandom state generators (PRSGs), pseudorandom function-like state generators (PRFSGs), and pseudorandom unitaries (PRUs). In this paper, we show a full unitary oracle separation between PRFSGs and PRUs.
    
    The separation holds between the strongest state notion and the weakest unitary notion: even adaptively secure, quantum-accessible PRFSGs do not imply non-adaptively secure, forward-only PRUs, even when their implementations are allowed to be non-unitary and use an arbitrary number of ancillary qubits. This reveals a fundamental distinction between pseudorandomness for quantum states and for quantum unitaries.
    
    Our main technical idea is to view a candidate PRU construction with access to state generation oracles as a map from the underlying oracle states to implemented unitaries, and to study the derivatives of this map. These derivatives are inherently low rank, and we exploit this low-rank structure to distinguish the resulting unitaries from truly random ones. We believe this differential perspective may be useful for studying other structural questions about quantum states and unitaries.
\end{abstract}
\ifnum\llncs=1
\fi
\ifnum\fullpage=1
\vfill

\paragraph{AI use disclosure.} The author first learned about the Poinca\'re inequality through ChatGPT 5.4 Pro. The mathematical ideas and proofs in this work were developed and completed by the author. ChatGPT 6 Astra was then used to prepare an initial written draft of the full proofs. The author subsequently rewrote and edited the manuscript. The author takes full responsibility for all mathematical claims and the content of this work.
\fi

\ifnum\submission=1
\else
	\clearpage
	\newpage
	\setcounter{tocdepth}{2}
	\tableofcontents
	\newpage
\fi

\section{Introduction}
\label{sec:introduction}

Quantum pseudorandomness extends computational randomness
from classical strings and functions to intrinsically quantum objects.
Ji, Liu, and Song introduced pseudorandom state generators (PRSGs) and
pseudorandom unitaries (PRUs)~\cite{C:JiLiuSon18}.
A PRSG uses a short classical key to prepare a quantum state whose
polynomially many copies are computationally indistinguishable from
copies of a Haar-random state.
A PRU, on the other hand, implements a keyed unitary, or possibly
near-unitary, transformation that is computationally indistinguishable
from a Haar-random unitary under quantum oracle queries.

PRSGs and PRUs do not have immediate counterparts in the classical
theory of pseudorandomness, as they concern two different types of
quantum objects: states and transformations.
Nevertheless, there is a useful analogy with the classical hierarchy.
A PRSG resembles a pseudorandom generator (PRG) in that a short
classical seed selects a single pseudorandom object, whereas a PRU
resembles a pseudorandom permutation (PRP) or a pseudorandom function (PRF) in that a short key specifies
a pseudorandom transformation over a large space.
In this analogy, pseudorandom function-like state generators (PRFSGs)~\cite{C:AnaQiaYue22}
may be closer to pseudorandom functions (PRFs): for a fixed key, they
associate a pseudorandom object with each classical input.
This analogy is imperfect, but it has generated rich cryptographic
applications.
Starting from PRSGs, quantum pseudorandomness has been used to construct private-key quantum money~\cite{C:JiLiuSon18}, commitments, encryption, and secure multiparty computation~\cite{C:AnaQiaYue22,C:MorYam22,TCC:AGQY22,C:BCKM21}. 
PRFSGs further support applications such as secret-key encryption and message authentication~\cite{C:AnaQiaYue22}.
PRUs have also found applications in
unclonable encryption~\cite{ITCS:PorembaRagavanVaikuntanathan26}.

Classically, PRGs, PRFs, and PRPs are tightly connected to each other.
The GGM construction turns PRGs into PRFs~\cite{JACM:GGM86}, while the Luby-Rackoff construction turns PRFs into PRPs~\cite{SICOMP:LR88}.
Thus, at the level of existence, these notions are equivalent.
The quantum landscape is considerably more fragmented. At a basic level, quantum pseudorandom primitives can exist relative to oracles where classical pseudorandomness does not exist \cite{Kre21}, even with classical oracles \cite{STOC:KQST23,STOC:KreQiaTal25} or even allowing negligible errors \cite{BEHMV26,Bar25}.

Even within quantum pseudorandomness, the picture remains fragmented.
Constructions of PRSGs and PRFSGs from classical pseudorandomness are well-established
\cite{C:JiLiuSon18,C:BraShm20,C:AnaQiaYue22,TCC:AGQY22}.
However, constructions of PRUs from classical pseudorandomness were obtained much later
and required substantially new techniques
\cite{FOCS:MPSY24,STOC:MaHua25,FLM+26}, while their equivalence with PRS(F)Gs remains unclear.
Several intermediate notions were introduced and constructed to capture pseudorandom operations on quantum inputs in the meantime, such as
pseudorandom state scramblers~\cite{TCC:LQSYZ24} and pseudorandom isometries~\cite{EC:AGKL24}.

This raises a natural question: are these notions of quantum pseudorandomness equivalent at the level of existence, as their classical counterparts are? In particular, can PRFSGs be used to construct PRUs?
There is an apparent obstacle: a PR(F)SG, or state generation more generally, only needs to operate on a low-dimensional subspace to generate individual pseudorandom states, whereas a PRU must specify a coherent transformation on an entire Hilbert space.
This dimensional mismatch has led to partial oracle separations \cite{BEHMV26,EPRINT:GLMY25} for PRU constructions with severely bounded ancilla space. 
However, these results leave open whether the ancilla restriction is essential.
Since general PRU constructions may use polynomially many ancillary qubits, these results leave open a full oracle separation with unrestricted ancilla space.

\subsection{Our result}

We give a unitary oracle separation between PRFSGs and PRUs.
\begin{theorem}
There is a unitary oracle relative to which PRFSGs exist, whereas PRUs do not.
\end{theorem}
This separation holds for a strong notion of PRFSGs---adaptively secure,
quantum-accessible PRFSGs---and a weak notion of PRUs---non-adaptively secure,
forward-only PRUs. Consequently, the same oracle separates any weaker notion of PRFSGs from any stronger notion of PRUs; for example, it separates classically accessible PRFSGs from adaptively secure PRUs with inverse and conjugate access \cite{C:Zhandry25}.
Furthermore, the impossibility holds even for candidate implementations that are not themselves unitary: they may be implemented by arbitrary efficient quantum channels using polynomially many ancillary qubits and intermediate measurements.
That is, our separation even rules out PRU implementations with negligible errors.

Our separation may also be related to the distinction between state
synthesis and unitary synthesis. Arbitrary quantum states can be
prepared efficiently using one query to a suitable classical oracle
\cite{SODA:Rosenthal24}, whereas an analogous one-query
cannot implement arbitrary unitaries \cite{STOC:LomMaWri24}. While the connection is unclear, these results similarly highlight a distinction between preparing quantum states
and implementing unitary transformations.

Our proof analyzes the derivatives of the oracle PRU construction by interpreting it as a map from oracles (or the states defining oracles) to the space of quantum operations. 
In this interpretation, the observation that ``the state generation can be implemented on a low-dimensional subspace'' is captured by the low rank of the gradient, which in turn allows us to show that the norm of the gradient is relatively small on average. 
Intuitively, the smallness of the derivative makes the output of the map---the resulting quantum operator---insensitive to perturbations, suggesting a concentration phenomenon.
A Poincaré inequality formally converts this gradient bound into concentration of the resulting quantum operation.
Such concentration is incompatible with the behavior of a random unitary, yielding the separation.
To handle non-unitary implementations, we establish a dichotomy: for a suitable input state, either the candidate produces a noticeably mixed output, which already distinguishes it from a unitary, or its output falls into the same concentration regime as above.

\subsection{Technical overview}
We use the common-Haar function-like state (CHFS) oracle together with the $\QPSPACE$ oracle \cite{TCC:AnaGulLin24,BEHMV26,EPRINT:GLMY25}.
In the overview, we use the CHFS oracle with a fixed length, say, $d=\lfloor\sqrt{\lambda}\rfloor$ for simplicity.
The CHFS oracle is defined by a family of $d$-qubit pure states $\Phi=(\ket{\phi_x})_{x\in \bit^d}$.
Given this, the CHFS oracle query is processed by
\[
S_d^\Phi = \sum_x \ketbra x \otimes S_{\phi_x}
\]
in superposition, 
where $S_\phi$ swaps between
$\ket{0^{d+1}}$ and $\ket{\phi,1}$ and is given by
\[
    S_\phi
    \defeq
    \Id
    -
    \proj{0^{d+1}}
    -
    \proj{\phi,1}
    +
    \ketbra{0^{d+1}}{\phi,1}
    +
    \ketbra{\phi,1}{0^{d+1}}
\]
which is inspired by the reflection/swap Haar state oracle~\cite{C:GolZha25,EC:BMMMY25}.

Let $\mathcal G^{S_d^\Phi}$ be a candidate PRU.
We focus on simulating the application of $\mathcal G^{S_d^\Phi}$ on a specific state \emph{without using} $S_d^\Phi$. In the full proof, we use process tomography for logarithmically small $d$ in the simulation. 
Given the simulation, the quantum OR test \cite{SODA:HarLinMon17} over the key together with a strong complexity oracle, say $\QPSPACE$, shows the insecurity of $\mathcal G^{S_d^\Phi}$ following, e.g., \cite{EC:CheColSat25}.
We omit the key in the algorithm as it is less relevant to the simulation.

\paragraph{Identity replacement and failures.}
We briefly explain the attempts to separate PRFSGs and PRUs in \cite{BEHMV26,EPRINT:GLMY25}. 
The main idea was to replace the oracle query $S_d^\Phi$ by the identity $I$, i.e., to consider $\mathcal G^{I}$ and then apply this PRU candidate on a Haar random state or the first half of the maximally entangled state $\ket{\Omega}$.

We consider the random-state input for intuition for now.
Without an ancillary register, an oracle-independent unitary maps random states into random states. 
Also note that the oracle query $S_d^\Phi$ only perturbs a low-dimensional subspace spanned by $\ket{x,0^{d+1}}$ and $\ket{x,\phi_x,1}$, thus for random state input each query does not change the state much. Therefore, replacing the oracle by the identity does not change the final output state much.

It turns out that ancillary registers break precisely this argument.
With a sufficiently large ancilla, an oracle-independent unitary can move the entire input state into the ancilla.
For example, for the query register $Q$ and the ancillary register $A$, it may map
\[
\ket{\psi}\ket{0}_A
\longmapsto
\ket{0^d,0^{d+1}}_Q\ket{\psi}_A
\begin{cases}
\xmapsto{S_d^\Phi}
\ket{0^d,\phi_{0^d},1}_Q\ket{\psi}_A\\
\xmapsto{~I~}
\ket{0^d,0^{d+1}}_Q\ket{\psi}_A
\end{cases}
\]
where we omit the unrelated registers.
Thus, replacing $S_d^\Phi$ by $I$ can drastically change an intermediate state.
This alone does not imply a large difference in the induced operation, since the discrepancy may eventually disappear into the ancillary workspace.
However, it breaks the original argument, as the two intermediate states remain far even after averaging over a random input.

Worse, there is a more severe counterexample: replacing the CHFS oracle
by the identity may result in a completely different procedure.
Fix $x\neq 0^d$. We use the CHFS oracle first on input $0^d$ and then on input $x$, uncomputing the input register after the queries.
We then apply an arbitrarily chosen unitary $V$, controlled on whether the
$d$-qubit state register in the ancilla is $\ket{0^d}$:
\[
\ket{\psi}\ket{0^{d+1}}_A
\begin{cases}
    \xmapsto{S_d^\Phi}
    \ket\psi\ket{\phi_{0^d},1}_A
    \simmapsto{S_d^\Phi}
    \ket\psi\ket{\phi_{0^d},1}_A
    \simmapsto{C_{0^d}V}
    \ket\psi\ket{\phi_{0^d},1}_A
    \\
    \xmapsto{~I~}
    \ket\psi\ket{0^d,0}_A\,\,
    \xmapsto{~I~}
    \ket\psi\ket{0^d,0}_A\,\,
    \xmapsto{C_{0^d}V}
    V\ket\psi\ket{0^d,0}_A
\end{cases}
\]
where in the above equation the first approximation is due to the almost orthogonality of $\ket{\phi_{0^d}}$ and $\ket{\phi_x}$ and the second is due to the almost orthogonality between $\ket{\phi_{0^d}}$ and $\ket{0^d}$.
This gives a more serious obstruction: the purity is almost unaffected, while replacing the CHFS oracle by $I$ can change the induced operation by an arbitrary unitary $V$. Note that the difference in the final qubit is also a problem.



\paragraph{Random CHFS replacement and derivatives.}
Our starting point is simple: we replace the original CHFS oracle $S_d^\Phi$
by another CHFS oracle $S_d^{\Phi'}$ defined by an independent family of random
states $\Phi'=(\ket{\phi'_x})_{x\in\bit^d}$.
This choice circumvents the second counterexample, as the same near-orthogonality
properties still hold for the new CHFS oracle.
Still, the first problem persists: we need a new proof strategy to compare
the two cases.

We use the following perspective.
Consider the output state
\[
\ket{\rho_\Phi}=(\mathcal G^{S_d^\Phi}\otimes I)\ket\Omega
\]
where, for simplicity, we suppose that $\mathcal G^{S_d^\Phi}$
is implemented by a unitary circuit with no ancilla, and identify it with the resulting unitary. The general case will be discussed later.
We then interpret this state as the output of the map
\[
(\text{product space of states})
\longrightarrow
(\text{space of output states}),
\qquad
\Phi
\longmapsto
\ket{\rho_{\Phi}}.
\]
This is obviously continuous, and moreover is a differentiable function! Here a derivative comes in, which exhibits the low-dimensionality of the CHFS oracle in a different form.

We use dots to denote differentials; e.g., $\dot\phi$ and $\dot S_\phi$ denote derivatives of $\phi$ and $S_\phi$,
respectively.
For ease of exposition, we suppress tangent coordinates and other
mathematical details, and also use dot notation for gradients.
To see how the differential viewpoint captures the low-rank intuition, recall
\[
S_\phi
\defeq
\Id
-
\proj{0^{d+1}}
-
\proj{\phi,1}
+
\ketbra{0^{d+1}}{\phi,1}
+
\ketbra{\phi,1}{0^{d+1}}.
\]
A direct calculation shows that its derivative is
\[
\dot S_\phi
=
-\ketbra{\dot\phi,1}{\phi,1}
-\ketbra{\phi,1}{\dot\phi,1}
+\ketbra{0^{d+1}}{\dot\phi,1}
+\ketbra{\dot\phi,1}{0^{d+1}}
\]
where the identity $I$ vanishes. 
In particular, the derivative has a constant rank, $\rank(\dot S_\phi)=O(1)$.
Accordingly, the derivative of the CHFS oracle becomes
\[
\dot S_d^\Phi = \sum_{x \in \bit^d}
\ketbra{x}\otimes \dot S_{\phi_x},
\]
which has low rank, $O(2^d)$, compared with the query register dimension $O(2^{2d})$.

Write a $q$-query oracle algorithm $\mathcal G^{S_d^\Phi}$ as
\[
W^\Phi
=
U_q S_d^\Phi U_{q-1} S_d^\Phi
\cdots
U_1 S_d^\Phi U_0
\]
by deferring all intermediate measurements.
Differentiating it gives
\begin{align*}
   \dot 
W^\Phi = {} &U_q \dot S_{d}^{\Phi} U_{q-1}S_d^\Phi
\cdots U_1S_d^\Phi U_0
\\
&+
U_q S_d^\Phi U_{q-1}\dot S_{d}^{\Phi}
\cdots U_1S_d^\Phi U_0
\\
&+\cdots+
U_q S_d^\Phi U_{q-1}S_d^\Phi
\cdots U_1\dot S_{d}^{\Phi}U_0.
\end{align*}
Thus, the derivative of the oracle algorithm has rank at most an
$O(q/2^d)$ fraction of the ambient dimension.
We exploit this low relative rank to bound the differential of the
output state obtained by applying the algorithm to half of the
maximally entangled state $\ket{\Omega}$. 
Intuitively, this controls how much the output state
$\ket{\rho_{\Phi}}$ changes as we perturb the oracle
states $\Phi$.

What we actually show is the following squared norm bound 
\[
\left\|
\ketbra{\dot\rho_{\Phi}}{\rho_{\Phi}}
+
\ketbra{\rho_{\Phi}}{\dot\rho_{\Phi}}
\right\|^2
\lesssim q^2.
\]
Note that this is a derivative of the density matrix $\ketbra{\rho_{\Phi}}=:\rho_{\Phi}$ where $\ket{\dot\rho_\Phi}
=
(\dot W^\Phi\otimes I)\ket{\Omega}.$
This turns out to suffice for our purpose.
The Poincar\'e inequality then roughly says that
\[
\EE_{\Phi_d}
\normtwo{\rho_\Phi-\EE_{\Phi}[\rho_\Phi]}^2
\leq
\frac{\EE_\Phi\!\left[\|\ketbra{\dot\rho_{\Phi}}{\rho_{\Phi}}+\ketbra{\rho_{\Phi}}{\dot\rho_{\Phi}}\|^2\right]}
{2^{d+1}-1}
=O(q^2 2^{-d}).
\]
Thus, the output state concentrates around the averaged state
$\EE_\Phi[\rho_\Phi]$, which can be effectively simulated.

\paragraph{Dealing with ancilla and non-unitarity.}
We now explain why this approach can handle ancillary qubits; it is convenient
to discuss the non-unitary case at the same time.
Recall that $W^\Phi
=
U_q S_d^\Phi U_{q-1} S_d^\Phi
\cdots
U_1 S_d^\Phi U_0$
represents the unitary part of the algorithm $\mathcal G^{S_d^\Phi}$.
Let $A$, $R$, and $B$ denote the algorithm register, the reference register
of the maximally entangled state, and the ancillary register, respectively.
We write
\[
\ket{\Xi_\Phi}_{ABR}
\defeq
(W^\Phi\otimes\Id_R)
\bigl(\ket{\Omega}_{AR}\ket{0}_B\bigr),
\qquad
\rho_\Phi
\defeq
\Tr_B\proj{\Xi_\Phi}.
\]
The squared norm that we need to bound is therefore
\[
\normtwo{
\Tr_B\left(
\ketbra{\dot\Xi_\Phi}{\Xi_\Phi}
+
\ketbra{\Xi_\Phi}{\dot\Xi_\Phi}
\right)
}^2.
\]
Here $\ket{\dot\Xi_\Phi}$ is generated by the low-rank derivative
$\dot W^\Phi$.

Suppose first that $\rho_\Phi$ is nearly pure, otherwise the purity test can distinguish $\mathcal G$ from random unitaries.
We observe that this purity condition gives the decomposition
\[
\ket{\Xi_\Phi}_{ABR}
=
\ket{\Xi_\Phi^{\mathrm{id}}}_{AR}
\otimes
\ket{\chi_\Phi}_B
+
\ket{e_\Phi},
\]
for a small error $\ket{e_\Phi}$.
Define
\[
\ket{\eta_\Phi}_{AR}
\defeq
(\Id_{AR}\otimes\bra{\chi_\Phi}_B)
\ket{\dot\Xi_\Phi}_{ABR}.
\]
Then the derivative of the reduced state decomposes as
\[
\ketbra{\eta_\Phi}{\Xi_\Phi^{\mathrm{id}}}
+
\ketbra{\Xi_\Phi^{\mathrm{id}}}{\eta_\Phi}
+
E_\Phi, \qquad E_\Phi
\defeq
\Tr_B\!\left(
\ketbra{\dot\Xi_\Phi}{e_\Phi}
+
\ketbra{e_\Phi}{\dot\Xi_\Phi}
\right).
\]
The first two terms have the same form as in the pure case, while
the last term is small due to the smallness of
$\ket{e_\Phi}$ together with our derivative bound.

\ifnum\fullpage=0
\paragraph{AI use disclosure.} The author first learned about the Poinca\'re inequality through ChatGPT 5.4 Pro. The mathematical ideas and proofs in this work were developed and completed by the author. ChatGPT 6 Astra was then used to prepare an initial written draft of the full proofs. The author subsequently rewrote and edited the manuscript. The author takes full responsibility for all mathematical claims and the content of this work.
\fi

\section{Preliminaries}
\label{sec:preliminaries}

\paragraph{Notation.}
All functions and algorithms are parameterized by
the security parameter $\lambda\in\N$.
We write $[m]=\{1,\ldots,m\}$.
For a distribution $\mathcal D$ and a finite set $S$, we write
$x\gets\mathcal D$ and $x\gets S$ for sampling from $\mathcal D$ and uniform sampling from $S$, respectively.

We identify each quantum register with a finite-dimensional
Hilbert space.
We write $\cL(\cH)$ for the space of linear operators on $\cH$
and $\Herm(\cH)$ for its real vector space of Hermitian operators.
On $\Herm(\cH)$, the Hilbert--Schmidt inner product is
$(X,Y)_{\rm HS}=\Tr(XY)$.

A quantum state on a register $A$ is represented
by a density matrix $\rho\geq0$ with $\Tr(\rho)=1$.
A pure state is represented by a unit vector $\ket\psi$, and we
write $\ketbra\psi$ for its density matrix.

A quantum channel $\mathcal E$ from $A$ to $A'$ is a completely
positive, trace-preserving linear map; $\mathcal E(\rho)$ denotes
its output on input $\rho$.
We write $\mathsf U(A)$ for the unitaries on $A$.
A unitary $U$ induces the channel $\mathcal U(\rho)=U\rho U^\dagger$.
The identity unitary and channel on $A$ are denoted by $\Id_A$
and $\mathcal I_A$, respectively.

\subsection{Distances and tomography}
For a vector $v$, we write $\norm v$ for its Euclidean norm.
For an operator $X$, the trace, Hilbert--Schmidt, and operator
norms are
\[
\normone X=\Tr\sqrt{X^\dagger X},
\qquad
\normtwo X=\sqrt{\Tr(X^\dagger X)},
\qquad
\normop X=\sup_{\norm v=1}\norm{Xv}.
\]
We use $\normtwo X\leq\normone X$ and
$\normone{YXZ}\leq\normop Y\normone X\normop Z$
for operators of compatible dimensions.

The trace distance between $\rho$ and $\sigma$ is
$\normone{\rho-\sigma}/2$.
For any two-outcome measurement, the absolute difference
in acceptance probabilities on $\rho$ and $\sigma$ is at most
their trace distance, and some measurement achieves this bound.
Applying the same channel to two states cannot increase their
trace distance. 
For pure states $\ket\psi,\ket\phi$,
\[
\frac{
\normone{\proj\psi-\proj\phi}}{2}
=\sqrt{1-\abs{\braket{\psi}{\phi}}^2}=
\frac{\normtwo{\proj\psi-\proj\phi}}{\sqrt 2}.
\]
For two (possibly unnormalized) states $\ket u,\ket e\in A\otimes B$, we also have the inequality
$\normtwo{\Tr_B(\ketbra{u}{e})}\leq\norm{\ket u}\norm{\ket e}$.

For channels $\mathcal E,\mathcal F$ from $A$ to $A'$, their
diamond-norm distance is
\[
\normdiamond{\mathcal E-\mathcal F}
=\sup_{C,\rho_{AC}}
\normone{
(\mathcal E\otimes\mathcal I_C)(\rho_{AC})
-(\mathcal F\otimes\mathcal I_C)(\rho_{AC})
},
\]
where the supremum is over all finite-dimensional $C$ and all states $\rho_{AC}$ on $A\otimes C$.
Replacing gates in a quantum computation changes its output state,
in trace norm, by at most the sum of the diamond-norm distances
of the replaced gates.

We use the following process tomography result from \cite{HKOT23}.

\begin{theorem}
\label{thm:process-tomography}
Given oracle access to a unitary $Z$ in dimension $D$ and
$\zeta,\delta\in(0,1)$, an algorithm outputs a classical description
of a unitary $\widetilde Z$ such that
$\normdiamond{\mathcal Z-\widetilde{\mathcal Z}}\leq\zeta$
with probability at least $1-\delta$.
It uses $O(D^2\zeta^{-1}\log(2/\delta))$ queries and time
polynomial in $D,1/\zeta,\log(2/\delta)$.
\end{theorem}

\subsection{Haar randomness and concentration}
Haar-random pure states and unitaries are sampled with respect
to the normalized Haar measure in the relevant dimension.
We sometimes write $\psi\gets\Haar$ for sampling from the Haar distribution.

The following concentration bound follows from \cite{Mec19,Kre21}.

\begin{lemma}
\label{cor:haar-concentration}
Let $U_1,\ldots,U_m$ be independent Haar-random unitaries
of dimensions $d_1,\ldots,d_m\geq d$.
Let $g(U)$ be the acceptance probability of a quantum algorithm
making at most $q\geq1$ queries to $\bigoplus_{j=1}^m U_j$.
Then, for every $t>0$,
\[
\Prb[g(U)\geq\EE g(U)+t]
\leq\exp\left(-\frac{(d-2)t^2}{24q^2}\right).
\]
The same bound holds for $\Prb[g(U)\leq\EE g(U)-t]$.
\end{lemma}

We also use the Chernoff-Hoeffding bound.
\begin{lemma}
\label{lem:bernoulli-chernoff}
Let $X_1,\ldots,X_r$ be independent Bernoulli variables with
$\EE X_i=p$. For $0<p<a<1$,
\[
\Prb\left[\frac1r\sum_{i=1}^rX_i\geq a\right]
\leq
\exp\left(-r\left[
a\log\frac ap+(1-a)\log\frac{1-a}{1-p}
\right]\right).
\]
\end{lemma}

\subsection{Choi states and purity tests}
Let $A$ have dimension $d_A$, and let $R$ be an ancillary
register of the same dimension.
The maximally entangled state and the Choi
state of a channel $\mathcal E$ on $A$ are
\[
\ket\Omega_{AR}
\defeq\frac1{\sqrt{d_A}}\sum_{a=1}^{d_A}\ket a_A\ket a_R,
\qquad
J(\mathcal E)\defeq
(\mathcal E\otimes\mathcal I_R)(\proj\Omega).
\]
Thus $J(\mathcal E)$ is prepared by applying $\mathcal E$
to the $A$ register of $\ket\Omega_{AR}$.
If $\mathcal V$ is induced by a unitary $V$, then
$J(\mathcal V)$ is the pure state with state vector
$(V\otimes\Id_R)\ket\Omega$.

We use the swap test and the purity test.

\begin{lemma}[Swap and purity tests]
\label{lem:swap-test}
The swap test on $\rho\otimes\sigma$ passes with probability
$(1+\Tr(\rho\sigma))/2$.
On two independent copies of $\rho$, it is the
purity test and fails with probability
$(1-\Tr(\rho^2))/2$.
\end{lemma}

We will use the following lemmas for almost pure states.
\ifnum\submission=1
The proof of these lemmas can be found in \cref{app:prelproofs}.

\begin{lemma}[Purification near a pure state]
\label{lem:elementary-quantum-estimates}
Let $\ket\Xi_{AB}$ be a purification of $\rho_A$.
If $\normone{\rho-\proj\psi}\leq\eps\leq2$ for a pure state
$\ket\psi_A$, then some unit vector $\ket\eta_B$ satisfies
$\norm{\ket\Xi-\ket\psi\ket\eta}\leq\sqrt\eps$.
\end{lemma}

\begin{lemma}[Trace norm near pure states]
\label{lem:rank-one-trace-norm-comparison}
Let $\rho,\sigma$ be density operators and $P,Q$ rank-one projectors
on the same space. If $\epsilon_\rho=\normone{\rho-P}$ and
$\epsilon_\sigma=\normone{\sigma-Q}$, then $
\normone{\rho-\sigma}
\leq(1+\sqrt2)(\epsilon_\rho+\epsilon_\sigma)
+\sqrt2\normtwo{\rho-\sigma}.$
\end{lemma}

For a channel $\mathcal E$ on $A$, define its Choi error and
purity defect by
\[
e(\mathcal E)
\defeq
\inf_{V\in\mathsf U(A)}
\normone{J(\mathcal E)-J(\mathcal V)},
\qquad
\Delta(\mathcal E)
\defeq
1-\Tr\bigl(J(\mathcal E)^2\bigr).
\]
Given two forward calls to $\mathcal E$, we can test the purity of its Choi state: prepare two copies
of $J(\mathcal E)$ and apply the purity test.
It fails with probability $\Delta(\mathcal E)/2$;
for a unitary channel, it never fails.
The following lemma relates this probability to the Choi error.

\begin{lemma}[Purity and Choi error]
\label{lem:purity-to-unitary-choi}
Every channel $\mathcal E$ from $A$ to itself satisfies
\[
\Delta(\mathcal E)
\leq e(\mathcal E)
\leq 5\sqrt{\Delta(\mathcal E)}.
\]
\end{lemma}
\else

\begin{lemma}[Purification near a pure state]
\label{lem:elementary-quantum-estimates}
Let $\ket\Xi_{AB}$ be a purification of $\rho_A$.
If $\normone{\rho-\proj\psi}\leq\eps\leq2$ for a pure state
$\ket\psi_A$, then some unit vector $\ket\eta_B$ satisfies
$\norm{\ket\Xi-\ket\psi\ket\eta}\leq\sqrt\eps$.
\end{lemma}

\begin{proof}
Let $w=(\bra\psi\otimes\Id_B)\ket\Xi$ and
$p=\norm w^2=\Tr(\proj\psi\rho)$.
The measurement $\{\proj\psi,\Id-\proj\psi\}$ accepts
$\proj\psi$ with probability one and $\rho$ with probability $p$.
Hence $1-p\leq\normone{\rho-\proj\psi}/2\leq\eps/2$.
If $p>0$, let $\ket\eta=w/\sqrt p$; then
$\norm{\ket\Xi-\ket\psi\ket\eta}^2=2-2\sqrt p\leq\eps$.
If $p=0$, then $\eps=2$ and any unit vector $\ket\eta$ suffices.
\end{proof}

\begin{lemma}[Trace norm near pure states]
\label{lem:rank-one-trace-norm-comparison}
Let $\rho,\sigma$ be density operators and $P,Q$ rank-one projectors
on the same space. If $\epsilon_\rho=\normone{\rho-P}$ and
$\epsilon_\sigma=\normone{\sigma-Q}$, then $
\normone{\rho-\sigma}
\leq(1+\sqrt2)(\epsilon_\rho+\epsilon_\sigma)
+\sqrt2\normtwo{\rho-\sigma}.$
\end{lemma}

\begin{proof}
The triangle inequality and
$\normone{P-Q}=\sqrt2\normtwo{P-Q}$ give
\[
\normone{\rho-\sigma}
\leq\epsilon_\rho+\epsilon_\sigma+\sqrt2\normtwo{P-Q}.
\]
Also, $\normtwo{P-Q}\leq
\epsilon_\rho+\normtwo{\rho-\sigma}+\epsilon_\sigma$.
Combining the two inequalities proves the claim.
\end{proof}

For a channel $\mathcal E$ on $A$, define its Choi error and
purity defect by
\[
e(\mathcal E)
\defeq
\inf_{V\in\mathsf U(A)}
\normone{J(\mathcal E)-J(\mathcal V)},
\qquad
\Delta(\mathcal E)
\defeq
1-\Tr\bigl(J(\mathcal E)^2\bigr).
\]
Given two forward calls to $\mathcal E$, we can test the purity of its Choi state: prepare two copies
of $J(\mathcal E)$ and apply the purity test.
It fails with probability $\Delta(\mathcal E)/2$;
for a unitary channel, it never fails.
The following lemma relates this probability to the Choi error.

\begin{lemma}[Purity and Choi error]
\label{lem:purity-to-unitary-choi}
Every channel $\mathcal E$ from $A$ to itself satisfies
\[
\Delta(\mathcal E)
\leq e(\mathcal E)
\leq 5\sqrt{\Delta(\mathcal E)}.
\]
\end{lemma}
\begin{proof}
Let $\rho=J(\mathcal E)$.
Choose a unitary channel $\mathcal V$ attaining the infimum in the definition of $e(\mathcal E)$.
Since $J(\mathcal V)$ is pure, the measurement
$\{J(\mathcal V),\Id-J(\mathcal V)\}$ accepts $J(\mathcal V)$
with probability one. Hence
$\Tr(J(\mathcal V)\rho)\geq
1-\normone{\rho-J(\mathcal V)}/2
=1-e(\mathcal E)/2$.
Since $\Tr(\rho^2)\geq\Tr(J(\mathcal V)\rho)^2$,
we have $\Delta(\mathcal E)\leq e(\mathcal E)$.

Conversely, let $p$ be the largest eigenvalue of $\rho$ and
$\ket\gamma$ a corresponding eigenvector.
Since $p\geq\Tr(\rho^2)=1-\Delta(\mathcal E)$,
$\normone{\rho-\proj\gamma}=2(1-p)\leq2\Delta(\mathcal E)$.
As $\Tr_A\rho=\Id_R/d_A$ and partial trace cannot increase
trace distance,
\[
\normone{\Tr_A\proj\gamma-\Id_R/d_A}
\leq\normone{\proj\gamma-\rho}\leq2\Delta(\mathcal E).
\]
Write $\ket\gamma$ in Schmidt form and choose a unitary $V$ such that
\[
\ket\gamma=\sum_{j=1}^{d_A}\sqrt{s_j}\ket{a_j}_A\ket{r_j}_R,
\qquad
(V\otimes\Id_R)\ket\Omega
=\frac1{\sqrt{d_A}}\sum_{j=1}^{d_A}\ket{a_j}_A\ket{r_j}_R.
\]
Then
\[
\norm{\ket\gamma-(V\otimes\Id_R)\ket\Omega}^2
=\sum_j\left(\sqrt{s_j}-\frac1{\sqrt{d_A}}\right)^2
\leq\sum_j\abs{s_j-1/d_A}\leq2\Delta(\mathcal E).
\]
The pure-state distance bound and the triangle inequality give
$e(\mathcal E)\leq2\Delta(\mathcal E)+2\sqrt{2\Delta(\mathcal E)}
\leq5\sqrt{\Delta(\mathcal E)}$,
where the last inequality uses $0\leq\Delta(\mathcal E)\leq1$.
\end{proof}
\fi

\subsection{PRFSGs and PRUs}
The main targets of this paper are the pseudorandom function-like state generators (PRFSGs) and pseudorandom unitaries (PRUs)~\cite{TCC:AGQY22,C:JiLiuSon18}. By default we use the definitions in the oracle world.

\begin{definition}[PRFSG]
\label{def:prfsg}
Let $\kappa,m,n$ be such that
$\kappa,m=\omega(\log\lambda)$.
We say an oracle QPT algorithm $\Gen^{\mathcal O}$ is a quantum-accessible, adaptively secure
\emph{pseudorandom function-like state generator} (PRFSG) if the following hold:
\begin{itemize}
    \item $\Gen^{\mathcal O}$ maps $(k,x)\in\{0,1\}^{\kappa(\lambda)} \times\{0,1\}^{m(\lambda)}$ to a fresh $n(\lambda)$-qubit state in a new register;
    \item For every oracle QPT adversary $\mathcal A$, 
    \[
        \abs{
        \Prb_{k\gets\{0,1\}^{\kappa(\lambda)}}
        [\mathcal A^{\mathcal O,\Gen_k^{\mathcal O}(\cdot)}=1]
        -
        \Prb_{G_{\rm Haar}}
        [\mathcal A^{\mathcal O,G_{\rm Haar}(\cdot)}=1]
        }
        \leq\negl(\lambda)
    \]
    where $G_{\rm Haar}$ samples an independent Haar state $\ket{\psi_x}$ for each $x$ at the beginning of the game. On each query, it applies the isometry $\ket{x}_X\mapsto\ket{x}_X\ket{\psi_x}_Y$. The algorithm $\mathcal A$ is allowed to query the oracle coherently and adaptively.
\end{itemize}
\end{definition}

For the pseudorandom unitaries, we allow the construction to be a channel, i.e., allow some errors as long as they are indistinguishable from random unitaries.
In the main theorem, our adversary will query PRU candidates non-adaptively only in the forward direction.

\begin{definition}[PRU]
\label{def:uniform-inverseless-pru-construction}
Let $n=\omega(\log\lambda)$. 
We say an oracle QPT algorithm $\mathcal G^{\mathcal O}$ is a \emph{pseudorandom unitary} (PRU) if the following hold:
\begin{itemize}
    \item For each $k\in\{0,1\}^\lambda$, $\mathcal G_k^{\mathcal O}$ is a channel from $n(\lambda)$ qubits to $n(\lambda)$ qubits.
    \item For every oracle QPT adversary $\mathcal A$, 
    \[
        \abs{
        \Prb_{k\gets\{0,1\}^{\lambda}}
        [\mathcal A^{\mathcal O,\mathcal G_k^{\mathcal O}}=1]
        -
        \Prb_{U\gets\mathsf U(2^{n(\lambda)})}
        [\mathcal A^{\mathcal O,\mathcal U}=1]
        }
        \leq\negl(\lambda).
    \]
\end{itemize}
\end{definition}

\subsection{The quantum OR test}
\label{subsec:imported-tools}
We use the unitary $\QPSPACE$ oracle defined as follows.
\begin{definition}
\label{def:qpspace-oracle}
The unitary oracle $\QPSPACE$ takes a quantum register $A$,
a classical Turing machine $M$, a time bound $t$ encoded
in binary, and an abort qubit.
If $M$ outputs a valid unitary circuit on $A$ within $t$ steps,
the oracle applies that circuit to $A$.
Otherwise, it flips the abort qubit.
The oracle preserves the description registers and acts
coherently on superpositions of $(M,t)$.
\end{definition}

Recall the quantum OR test \cite{SODA:HarLinMon17,EC:BMMMY25}.

\begin{lemma}[Quantum OR test]
\label{lem:quantum-or}
Let $\{M_i\}_{i=1}^N$ be two-outcome measurements on the same state
$\xi$, and let $0<a<1/2$ and $b>0$. There is a quantum OR test acting on one copy of $\xi$
with the following guarantees. If some $M_i$ accepts with probability
at least $1-a$, the test accepts with probability at least
$(1-a)^2/7$. If every $M_i$ accepts with probability at most $b$,
the test accepts with probability at most $4Nb$.
\end{lemma}

\ifnum\submission=1
We allow the measurements to make $\QPSPACE$ oracle queries. It is known that the quantum-or lemma can be executed in $\QPSPACE$ even when the measurements queries $\QPSPACE$, see \cite[Appendix A]{EC:CheColSat25}. See \cref{cor:qpspace-quantum-or} for a slightly more expanded discussion.

\else
We allow the measurements to make $\QPSPACE$ oracle queries.
The following corollary is implicit in the previous works, see e.g., \cite[Appendix A]{EC:CheColSat25}.

\begin{corollary}[QPSPACE implementation]
\label{cor:qpspace-quantum-or}
In \cref{lem:quantum-or}, suppose that $M_i$ is described by unitary QPT circuits with $\QPSPACE$ access that are generated uniformly from $i$ and a common polynomial-length
classical input. 
Suppose that $\xi$ consists of
polynomially many qubits and $\log N=\poly(\lambda)$.
Then the quantum OR test has a uniform QPT implementation
with access to $\QPSPACE$.
\end{corollary}
\begin{proof}
The unitary OR construction of \cite[Appendix~A]{EC:CheColSat25} can be executed using $\QPSPACE$.
The proof of \cite[Lemma~5.2]{EC:BMMMY25} extends this implementation to binary measurements given uniformly by polynomial-size circuits with $\QPSPACE$ access.
Since the family $\{M_i\}_{i=1}^N$ has such implementations, $\xi$ has polynomially many qubits, and $\log N=\poly(\lambda)$, this result gives the claimed
QPT implementation with access to $\QPSPACE$.
\end{proof}

\fi

\section{Hilbert-valued Functions and Poincar\'e Inequalities}
\label{sec:poincare}
This section introduces the mathematical background on the basic differential analysis of functions from product spheres to Hilbert spaces, and the Poincar\'e inequality for those functions.

Throughout this paper, every function differentiated is $C^{}$\footnote{The number $1$ here simultaneously denotes $C^1$ and the footnote index. A $C^1$ function is a differentiable function with a continuous derivative.}.
All derivatives are taken with respect to real
parameters and, for vector- or matrix-valued maps, are computed entrywise
in fixed bases.

\subsection{Functions on spheres and their derivatives}
We first review the notions and results about the functions on spheres.
Let $\Sphere(\C^D):=\{\psi\in\C^D:\norm\psi=1\}$ be the complex unit sphere.
We regard it as the $(2D-1)$-dimensional real unit sphere in $\R^{2D}$,
with real inner product $\operatorname{Re}\langle u,v\rangle$.
Its tangent space at $\psi\in\Sphere(\C^D)$ is
\[
    T_\psi\Sphere(\C^D)
    =
    \{v\in\C^D:\operatorname{Re}\langle\psi,v\rangle=0\}.
\]
If $e_1,\ldots,e_{D-1}$ is a complex orthonormal basis of 
the complex orthogonal complement $\psi^\perp:=
    \{v\in\C^D:\langle\psi,v\rangle=0\}$ of $\psi$,
then
\begin{align}
    \{i\psi,e_1,ie_1,\ldots,e_{D-1},ie_{D-1}\}
    \label{eq:real-tangent-basis}
\end{align}
forms a real orthonormal basis of $T_\psi\Sphere(\C^D)$.

We consider a function $f:\Sphere(\C^D)\to\R$.
At $\psi\in\Sphere(\C^D)$, the derivative of $f$ along a direction
$\phi\in T_\psi\Sphere(\C^D)$ is defined by
\[
    \frac{\partial f}{\partial\phi}(\psi)
    :=
    \left.\frac{\dd}{\dd t}f(\psi(t))\right|_{t=0},
\]
where $\psi(t)$ is any differentiable curve on $\Sphere(\C^D)$ with
$\psi(0)=\psi$ and
$\left.\frac{\dd}{\dd t}\psi(t)\right|_{t=0}=\phi$.
This quantity depends only on $\phi$, not on the choice of the curve.

For the basis chosen in \cref{eq:real-tangent-basis}, the
gradient of $f$ at $\psi$ is represented by
\[
    \nabla f(\psi)
    =
    \left(
        \frac{\partial f}{\partial i\psi}(\psi),
        \frac{\partial f}{\partial e_1}(\psi),
        \frac{\partial f}{\partial ie_1}(\psi),
        \ldots,
        \frac{\partial f}{\partial e_{D-1}}(\psi),
        \frac{\partial f}{\partial ie_{D-1}}(\psi)
    \right).
\]
We are mainly interested in the squared gradient norm
\[
    \|\nabla f(\psi)\|^2
    =
    \left|\frac{\partial f}{\partial i\psi}(\psi)\right|^2
    +
    \sum_{j=1}^{D-1}
    \left(
        \left|\frac{\partial f}{\partial e_j}(\psi)\right|^2
        +
        \left|\frac{\partial f}{\partial ie_j}(\psi)\right|^2
    \right),
\]
which is independent of the choice of the real orthonormal tangent basis.

We recall the Poincar\'e inequality on the complex sphere.
\begin{theorem}
[Poincar\'e inequality on the complex sphere]
\label{thm:complex-sphere-poincare}
Let $f:\Sphere(\C^D)\to\R$ be a $C^1$ function, and let
$\psi$ be uniformly distributed on $\Sphere(\C^D)$. Then
\[
    \Exp_{\psi}\left|f(\psi)-\Exp_{\psi} f(\psi)\right|^2
    =
    \Var(f(\psi))
    \leq
    \frac{1}{2D-1}\Exp_{\psi}\|\nabla f(\psi)\|^2.
\]
\end{theorem}

\subsection{Hilbert-valued functions on complex spheres}
We next extend the preceding definitions to functions taking values in a
finite-dimensional real Hilbert space $\mathcal K$.
Recall that a finite-dimensional real Hilbert space is a real inner-product space;
we write $\langle\cdot,\cdot\rangle_{\mathcal K}$ for the inner product of $\mathcal K$ and
$\|\cdot\|_{\mathcal K}$ for the induced norm.

Let $b_1,\ldots,b_m$ be an orthonormal basis of $\mathcal K$.
Then, for every $x,y\in\mathcal K$,
\[
x=
    \sum_{r=1}^m
    \langle b_r,x\rangle_{\mathcal K} b_r,
    \quad
\langle x,y\rangle_{\mathcal K}
    =
    \sum_{r=1}^m
    \langle b_r,x\rangle_{\mathcal K}
    \langle b_r,y\rangle_{\mathcal K},
    \quad
\|x\|_{\mathcal K}^2
    =
    \sum_{r=1}^m
    \left|\langle b_r,x\rangle_{\mathcal K}\right|^2.
\]
A map $F:\Sphere(\C^D)\to\mathcal K$ is said to be $C^1$ if its
coordinates in an orthonormal basis are $C^1$.
Expectations of $\mathcal K$-valued random variables are taken
coordinatewise.

For $\phi\in T_\psi\Sphere(\C^D)$, define its directional derivative by
\[
    \frac{\partial F}{\partial\phi}(\psi)
    :=
    \left.\frac{\dd}{\dd t}F(\psi(t))\right|_{t=0},
\]
where $\psi(t)$ is any differentiable curve on $\Sphere(\C^D)$ with $\psi(0)=\psi$ and $\left.\frac{\dd}{\dd t}\psi(t)\right|_{t=0}=\phi$. The coordinatewise differentiation is
\[
    \left\langle
        b_r,
        \frac{\partial F}{\partial\phi}(\psi)
    \right\rangle_{\mathcal K}
    =
    \frac{\partial F_r}{\partial\phi}(\psi) \qquad \text{where} \quad
    F_r(\psi):=\langle b_r,F(\psi)\rangle_{\mathcal K}.
\]

For the real orthonormal tangent basis in
\cref{eq:real-tangent-basis}, write
\[
    \nabla F(\psi)
    =
    \left(
        \frac{\partial F}{\partial i\psi}(\psi),
        \frac{\partial F}{\partial e_1}(\psi),
        \frac{\partial F}{\partial ie_1}(\psi),
        \ldots,
        \frac{\partial F}{\partial e_{D-1}}(\psi),
        \frac{\partial F}{\partial ie_{D-1}}(\psi)
    \right)
\]
and define the squared gradient norm\footnote{This is the squared Hilbert-Schmidt norm of the differential and is denoted by $\|\dd F_\psi\|_{\mathrm{HS}}^2$ in the literature.} by
\[
    \|\nabla F(\psi)\|_{\mathcal K}^2
    :=
    \left\|
        \frac{\partial F}{\partial i\psi}(\psi)
    \right\|_{\mathcal K}^2
    +
    \sum_{j=1}^{D-1}
    \left(
        \left\|
            \frac{\partial F}{\partial e_j}(\psi)
        \right\|_{\mathcal K}^2
        +
        \left\|
            \frac{\partial F}{\partial ie_j}(\psi)
        \right\|_{\mathcal K}^2
    \right).
\]
Again, the gradient norm is independent of the choice of basis.

For a $\mathcal K$-valued random variable $X$, we define its variance by
\[
    \Var(X):=\Exp\|X-\Exp X\|_{\mathcal K}^2.
\]
The Poincar\'e inequality can be naturally extended to this setting.
\begin{corollary}[Hilbert-valued Poincar\'e inequality]
\label{cor:hilbert-complex-sphere-poincare}
Let $F:\Sphere(\C^D)\to\mathcal K$ be a $C^1$ function, and let
$\psi$ be uniformly distributed on $\Sphere(\C^D)$. Then
\[
    \Exp_{\psi}\|F(\psi)-\Exp_{\psi} F(\psi)\|_{\mathcal K}^2
    =\Var(F(\psi))
    \leq
    \frac{1}{2D-1}
    \Exp_{\psi}\|\nabla F(\psi)\|_{\mathcal K}^2.
\]
\end{corollary}

\begin{proof}
Let $b_1,\ldots,b_m$ be an orthonormal basis of $\mathcal K$. Write
$F_r(\psi):=\langle b_r,F(\psi)\rangle_{\mathcal K}$, which is a
real-valued function on $\Sphere(\C^D)$.
By Parseval's identity,
\[
    \Exp_{\psi}\|F(\psi)-\Exp_{\psi}F(\psi)\|_{\mathcal K}^2
    =
    \sum_{r=1}^m
    \Exp_{\psi}
    \left|F_r(\psi)-\Exp_{\psi}F_r(\psi)\right|^2
\]
Applying \cref{thm:complex-sphere-poincare} to each $F_r$ gives
\begin{align}
    \Exp_{\psi}\|F(\psi)-\Exp_{\psi}F(\psi)\|_{\mathcal K}^2
    \leq
    \frac{1}{2D-1}
    \sum_{r=1}^m
    \Exp_{\psi}
    \|\nabla F_r(\psi)\|^2.
    \label{eq:hilbert-poincare-coordinate}
\end{align}
For any real orthonormal tangent basis $\{v_a\}_a$, Parseval's identity gives
\[
    \sum_{r=1}^m
    \left|
        \frac{\partial F_r}{\partial v_a}(\psi)
    \right|^2=
    \sum_{r=1}^m
    \left|
        \left\langle
            b_r,
            \frac{\partial F}{\partial v_a}(\psi)
        \right\rangle_{\mathcal K}
    \right|^2=
    \left\|
        \frac{\partial F}{\partial v_a}(\psi)
    \right\|_{\mathcal K}^2.
\]
Therefore, 
Summing the preceding identity over $a$ and substituting into
\cref{eq:hilbert-poincare-coordinate} proves the claim.
\end{proof}

\subsection{Hilbert-valued functions on products of complex spheres}

We now extend the preceding results to products of complex spheres.
Fix the dimensions $D_1,\ldots,D_L\geq 2$ and define the product sphere
\[
    \mathcal M
    :=
    \prod_{j=1}^L\Sphere(\C^{D_j}),
\]
equipped with the product of the uniform measures on the individual spheres.
We write a point in $\mathcal M$ as
$\Psi=(\psi_1,\ldots,\psi_L)$.

For each $j$, let
$\{v_{j,a}\}_{a=1}^{2D_j-1}$ be a real orthonormal basis of
$T_{\psi_j}\Sphere(\C^{D_j})$.
For a function $F:\mathcal M\to\mathcal K$, let
$\partial_{j,a}F(\Psi)$ denote the directional derivative obtained by
varying only $\psi_j$ in the direction $v_{j,a}$.
We write
\[
    \nabla_jF(\Psi)
    :=
    \left(
        \partial_{j,1}F(\Psi),
        \ldots,
        \partial_{j,2D_j-1}F(\Psi)
    \right)
\]
and define
\[
    \|\nabla_jF(\Psi)\|_{\mathcal K}^2
    :=
    \sum_{a=1}^{2D_j-1}
    \|\partial_{j,a}F(\Psi)\|_{\mathcal K}^2.
\]
As before, this quantity is independent of the choice of the real
orthonormal tangent basis.

We first establish the Poincar\'e inequality for real-valued functions on
the product space adapting the argument of \cite[Proposition~4.3.1]{BGL14}.

\begin{lemma}[Poincar\'e inequality on products of complex spheres]
\label{lem:scalar-product-sphere-poincare}
Let $f:\mathcal M\to\R$ be a $C^1$ function, where
$\Psi=(\psi_1,\ldots,\psi_L)$ is distributed according to the product
of the uniform measures on $\Sphere(\C^{D_1}),\ldots,\Sphere(\C^{D_L})$.
Then
\begin{align}
    \Var(f(\Psi))
    \leq
    \sum_{j=1}^L
    \frac{1}{2D_j-1}
    \Exp_{\Psi}\|\nabla_jf(\Psi)\|^2.
    \label{eq:scalar-product-sphere-poincare}
\end{align}
\end{lemma}

\begin{proof}
We first consider the case $L=2$, and write
$f=f(\psi_1,\psi_2)$. 
Define $g(\psi_1):=\Exp_{\psi_2}f(\psi_1,\psi_2).$
The law of total variance gives
\begin{align}
    \Var_{\psi_1,\psi_2}(f)
    =
    \Exp_{\psi_1}\Var_{\psi_2}(f\mid\psi_1)
    +
    \Var_{\psi_1}(g).
    \label{eq:two-sphere-variance-decomposition}
\end{align}
For every fixed $\psi_1$, applying
\cref{thm:complex-sphere-poincare} to the function
$\psi_2\mapsto f(\psi_1,\psi_2)$ gives
\begin{align}
    \Exp_{\psi_1}\Var_{\psi_2}(f\mid\psi_1)
    \leq
    \frac{1}{2D_2-1}
    \Exp_{\psi_1,\psi_2}
    \|\nabla_2f(\psi_1,\psi_2)\|^2.
    \label{eq:two-sphere-second-coordinate}
\end{align}
On the other hand, applying
\cref{thm:complex-sphere-poincare} to $g$ gives
\begin{align}
    \Var_{\psi_1}(g)
    \leq
    \frac{1}{2D_1-1}
    \Exp_{\psi_1}\|\nabla g(\psi_1)\|^2.
    \label{eq:two-sphere-first-coordinate}
\end{align}
For every tangent direction $v$ at $\psi_1$, differentiation commutes
with the expectation over $\psi_2$, so
\[
    \frac{\partial g}{\partial v}(\psi_1)
    =
    \Exp_{\psi_2}
    \frac{\partial f}{\partial v}(\psi_1,\psi_2).
\]
Therefore, Jensen's inequality gives
\[
    \|\nabla g(\psi_1)\|^2
    \leq
    \Exp_{\psi_2}
    \|\nabla_1f(\psi_1,\psi_2)\|^2.
\]
Combining this with
\cref{eq:two-sphere-variance-decomposition,eq:two-sphere-second-coordinate,eq:two-sphere-first-coordinate},
we obtain
\[
    \Var_{\psi_1,\psi_2}(f)
    \leq
    \frac{1}{2D_1-1}
    \Exp_{\psi_1,\psi_2}\|\nabla_1f\|^2
    +
    \frac{1}{2D_2-1}
    \Exp_{\psi_1,\psi_2}\|\nabla_2f\|^2.
\]
Iterating the same argument over the $L$ coordinates proves
\cref{eq:scalar-product-sphere-poincare}.
\end{proof}

The same inequality extends to Hilbert-valued functions.

\begin{theorem}[Hilbert-valued product-sphere Poincar\'e inequality]
\label{thm:hilbert-product-poincare}
For every $C^1$ function $F:\mathcal M\to\mathcal K$,
\begin{align}
    \Exp_{\Psi}\|F(\Psi)-\Exp_{\Psi}F(\Psi)\|_{\mathcal K}^2
    \leq
    \sum_{j=1}^L
    \frac{1}{2D_j-1}
    \Exp_{\Psi}\|\nabla_jF(\Psi)\|_{\mathcal K}^2.
    \label{eq:hilbert-product-poincare}
\end{align}
\end{theorem}


\begin{proof}
We use the variance-decomposition argument in the proof of
\cref{lem:scalar-product-sphere-poincare}, applying
\cref{cor:hilbert-complex-sphere-poincare} in place of
\cref{thm:complex-sphere-poincare}.

For $L=2$, the Hilbert-valued variance satisfies
\[
\Var_{\psi_1,\psi_2}(F)
=
\Exp_{\psi_1}\Var_{\psi_2}(F\mid\psi_1)
+
\Var_{\psi_1}(\Exp_{\psi_2}F).
\]
Apply \cref{cor:hilbert-complex-sphere-poincare} to
$\psi_2\mapsto F(\psi_1,\psi_2)$ for each fixed $\psi_1$,
and to $\psi_1\mapsto\Exp_{\psi_2}F(\psi_1,\psi_2)$.
Differentiation commutes with expectation, and Jensen's inequality gives
\[
\|\nabla_1\Exp_{\psi_2}F\|_{\mathcal K}^2
\leq\Exp_{\psi_2}\|\nabla_1F\|_{\mathcal K}^2.
\]
Combining these bounds proves
\cref{eq:hilbert-product-poincare} for $L=2$.
Iterating the same argument proves it for arbitrary $L$.
\end{proof}

\section{CHFS Oracle and Derivatives}
\label{sec:chfs-resampling}

We define the unitarized, or swap, common Haar-random function-like state
(CHFS) oracle and study its derivatives in the context of oracle
algorithms.  
We allow controlled access to the oracle by default and
suppress the control qubit throughout.

We recall the Leibniz rules for any directional
derivative $\partial$. Differentiation of vector- and operator-valued
functions is understood entrywise, and
\[
    \partial\ketbra{u}{v}
    =
    \ketbra{\partial u}{v}
    +
    \ketbra{u}{\partial v},\qquad
    \partial(AB)
    =
    (\partial A)B+A(\partial B),
    \partial(A^\dagger)
    =
    (\partial A)^\dagger,
\]
where $\ket{\partial u} = \partial\ket{u}$.
In particular, if $\ket{\xi}_{AB}$ is a differentiable purification of
$\rho=\Tr_B\proj{\xi}$, then $\partial\rho
    =
    \Tr_B\left(
        \ketbra{\partial\xi}{\xi}
        +
        \ketbra{\xi}{\partial\xi}
    \right).$
    
\subsection{The CHFS oracle}
\label{sec:unitary-chfs-oracle}

The CHFS oracle is parameterized by a family of states
$\Phi=(\Phi_d)_{d\in\N}$.  For every input $x\in \bit^d$ of length $d\in\N$, we assign a $d$-qubit pure state. Formally, let
\[
    \Phi_d
    =
    \bigl(\ket{\phi_{x}}\bigr)_{x\in\{0,1\}^d},
    \qquad
    \phi_{x}\in\Sphere(\C^{2^d}),
\]
where all states $\ket{\phi_{x}}$ are sampled independently from Haar
measure as an initialization, and are fixed throughout the oracle execution.
For $\phi\in\Sphere(\C^{2^d})$, we define the swap unitary associated with it by
\begin{align}
    S_\phi
    \defeq
    \Id
    -
    \proj{0^{d+1}}
    -
    \proj{\phi,1}
    +
    \ketbra{0^{d+1}}{\phi,1}
    +
    \ketbra{\phi,1}{0^{d+1}}
    \label{eq:one-label-chfs-unitary}
\end{align}

The CHFS oracle $S^{\Phi}$ parameterized by $\Phi=(\Phi_d)_{d\in\N}$ consists of the family of oracles $(S_d^{\Phi})_{d \in \N}$ where $S_d^{\Phi}$ acts on $(2d+1)$ qubits by
\[
    S_d^{\Phi}
    \defeq
    \sum_{x\in\{0,1\}^d}
    \proj x\otimes S_{\phi_{x}}.
\]

We regard $S_d^\Phi$ as an operator-valued function of the state tuple
\[
    \Phi=(\phi_x)_{x\in\{0,1\}^d}
    \in
    \prod_{x\in\{0,1\}^d}\Sphere(\C^{2^d}).
\]
Accordingly, the tangent space becomes
\[
    T_\Phi
    \left(
        \prod_{x\in\{0,1\}^d}\Sphere(\C^{2^d})
    \right)
    =
    \bigoplus_{x\in\{0,1\}^d}
    T_{\phi_x}\Sphere(\C^{2^d}).
\]
More precisely, for $\delta\in T_{\phi_x}\Sphere(\C^{2^d})$, let $\phi_x(s)$ be any
differentiable curve on $\Sphere(\C^{2^d})$ such that
$\phi_x(0)=\phi_x$ and
$\left.\frac{\dd}{\dd s}\phi_x(s)\right|_{s=0}=\delta$.
Let $\Phi(s)$ be obtained from $\Phi$ by replacing only $\phi_x$ with
$\phi_x(s)$. We define
\[
    \partial_{x,\delta}S_d^\Phi
    :=
    \left.
    \frac{\dd}{\dd s}S_d^{\Phi(s)}
    \right|_{s=0}.
\]
Note that the derivative $\partial_{x,\delta}S_d^\Phi$ is supported only on the
input-$x$ block
since
\[
    S_d^\Phi
    =
    \sum_{y\in\{0,1\}^d}
    \proj y\otimes S_{\phi_y}.
\]

Note also that because of the coherent access, the phase of each state vector $\ket\phi$ must be considered.
The phase direction is $i\phi$ and will contribute to the derivative.

\subsection{CHFS oracle derivatives}
\label{subsec:chfs-layer-derivatives}
This section studies the derivatives of the CHFS oracles and presents two inequalities.
This section focuses on a fixed input length $d$.
For every $x\in\{0,1\}^d$, we define for simplicity that
\[
    \ket{A_x}
    \defeq
    \ket x\ket{0^{d+1}},\quad
    \ket{B_x}
    \defeq
    \ket x\ket{\phi_x,1},\quad
    \ket{u_x}
    \defeq
    \frac{\ket{B_x}-\ket{A_x}}{\sqrt2}
\]
so that
\begin{align}
\label{eqn:Sphid}
    S^{\Phi}_d = I - 2\sum_{x\in \bit^d} \ketbra{u_x}.
\end{align}

For each $x$, choose a complex orthonormal basis
$e_{x,1},\ldots,e_{x,2^d-1}$ of $\phi_x^\perp$, which gives a real orthonormal basis
\[
    i\phi_x,
    e_{x,1},ie_{x,1},
    \ldots,
    e_{x,2^d-1},ie_{x,2^d-1}
\]
of $T_{\phi_x}\Sphere(\C^{2^d})$ as in \cref{eq:real-tangent-basis}.
We define the basis vectors
\[
    \ket{t_{x,j}}
    \defeq
    \ket x\ket{e_{x,j}}\ket1,
    \qquad
    j\in[2^d-1].
\]

The CHFS oracle $S_d^\Phi$ can be regarded as a function on $\Phi$, so that on $\phi_x$. The following lemma gives the partial derivatives of $S_d^\Phi$ with respect to $T_{\phi_x}\Sphere(\C^{2^d})$.

\begin{lemma}
\label{lem:explicit-chfs-derivatives}
For every $x\in\{0,1\}^d$ and $j\in[2^d-1]$,
\begin{align*}
    \partial_{x,e_{x,j}}S_d^\Phi
    &=
    -\sqrt2
    \left(
        \ketbra{u_x}{t_{x,j}}
        +
        \ketbra{t_{x,j}}{u_x}
    \right),
    \\
    \partial_{x,ie_{x,j}}S_d^\Phi
    &=
    i\sqrt2
    \left(
        \ketbra{u_x}{t_{x,j}}
        -
        \ketbra{t_{x,j}}{u_x}
    \right),
    \\
    \partial_{x,i\phi_x}S_d^\Phi
    &=
    i
    \left(
        \ketbra{B_x}{A_x}
        -
        \ketbra{A_x}{B_x}
    \right).
\end{align*}
Each derivative is zero outside the input-$x$ block.
\end{lemma}
\begin{proof}
Fix $x$ and suppress the subscript $x$.
Since $\ket A$ is independent of $\phi$, while
$\ket B=\ket x\ket{\phi,1}$, we have
$\partial_{x,\delta}\ket u
    =
    \frac{1}{\sqrt2}\ket x\ket{\delta,1}$
for every $\delta\in T_\phi\Sphere(\C^{2^d})$.
By \cref{eqn:Sphid},
\[
    \partial_{x,\delta}S_d^\Phi
    =
    -2\left(
        \ketbra{\partial_{x,\delta}u}{u}
        +
        \ketbra{u}{\partial_{x,\delta}u}
    \right) = -\sqrt{2}\left(
    \ketbra{x,\delta,1}{u}+\ketbra{u}{x,\delta,1}
    \right).
\]

Plugging $\delta=e_j$ and $\delta=ie_j$ give the first and second derivatives.
Finally,
$\partial_{x,i\phi}\ket u=i\ket B/\sqrt2$, and hence
\[
    \partial_{x,i\phi}S_d^\Phi
    =
    -i\sqrt2
    \left(
        \ketbra {B} {u}-\ketbra {u} {B}
    \right)
    =
    i\left(
        \ketbra {B} {A}-\ketbra {A} {B}
    \right),
\]
where we used $\ket u=(\ket B-\ket A)/\sqrt2$.
Each derivative is supported only on the input-$x$ block because only
the $x$-th summand of $S_d^\Phi$ depends on $\phi_x$.
\end{proof}

For each $x$, we define the relevant space for the CHFS oracles
\[
    E_x
    \defeq
    \operatorname{span}\{A_x,B_x\},
    \qquad
    T_x
    \defeq
    \operatorname{span}\{t_{x,1},\ldots,t_{x,2^d-1}\},
\]
which are orthogonal to each other. 
Let $\Pi_{E_x}$ and $\Pi_{T_x}$ denote the corresponding orthogonal
projections. Note that
\[
\Pi_{E_x}+\Pi_{T_x} \le \ketbra x \otimes I, \quad \sum_{x\in \bit^d} \Pi_{E_x} +
\sum_{x\in \bit^d} \Pi_{T_x} \le I
\]
so that the space $E_x\oplus T_x$ are mutually orthogonal for distinct $x$.

Observe that they are relatively low rank compared to the ambient space. This gives the following upper bounds of the CHFS derivatives. 
\ifnum\submission=1
The proofs are deferred to \cref{app:chfsproofs}.

\begin{lemma}[Two-sided CHFS derivative sum]
\label{lem:two-sided-chfs-sum}
Let $P,Q$ be orthogonal projections of rank
at most $r$ on the CHFS query register together with an arbitrary ancillary register. Then
\[
    \sum_{x\in\{0,1\}^d}
    \sum_{\delta}
    \normtwo{
        Q\bigl(
            \partial_{x,\delta}S_d^\Phi
            \otimes\Id
        \bigr)P
    }^2
    \leq
    9r,
\]
where, for each $x$, $\delta$ ranges over the
basis $\{i\phi_x,
    e_{x,1},ie_{x,1},
    \ldots,
    e_{x,2^d-1},ie_{x,2^d-1}\}$.
\end{lemma}

\begin{lemma}[One-sided CHFS derivative sum]
\label{lem:one-sided-chfs-sum}
Let $I$ be any register. For every vector $\xi$ on the CHFS query
register together with $I$,
\[
    \sum_{x\in\{0,1\}^d}
    \sum_{\delta}
    \norm{
        \bigl(
            \partial_{x,\delta}S_d^\Phi
            \otimes\Id_I
        \bigr)\xi
    }^2
    \leq
    5\cdot 2^d\norm\xi^2,
\]
where, for each $x$, $\delta$ ranges over the
basis $\{i\phi_x,
    e_{x,1},ie_{x,1},
    \ldots,
    e_{x,2^d-1},ie_{x,2^d-1}\}$.
\end{lemma}
\else
\begin{lemma}[Two-sided CHFS derivative sum]
\label{lem:two-sided-chfs-sum}
Let $P,Q$ be orthogonal projections of rank
at most $r$ on the CHFS query register together with an arbitrary ancillary register. Then
\[
    \sum_{x\in\{0,1\}^d}
    \sum_{\delta}
    \normtwo{
        Q\bigl(
            \partial_{x,\delta}S_d^\Phi
            \otimes\Id
        \bigr)P
    }^2
    \leq
    9r,
\]
where, for each $x$, $\delta$ ranges over the
basis $\{i\phi_x,
    e_{x,1},ie_{x,1},
    \ldots,
    e_{x,2^d-1},ie_{x,2^d-1}\}$.
\end{lemma}

\begin{proof}
We suppress the subscript $x$ in this proof; the distinct $x$'s are only
considered when we use the summation over $x\in\bit^d$.
By \cref{lem:explicit-chfs-derivatives} and the parallelogram identity, we have
\ifnum\fullpage=0
\begin{align}
    &\normtwo{
        Q\bigl(
            \partial_{e_j}S_d^\Phi\otimes\Id
        \bigr)P
    }^2
    +
    \normtwo{
        Q\bigl(
            \partial_{ie_j}S_d^\Phi\otimes\Id
        \bigr)P
    }^2
    \nonumber\\
    &\qquad=
    4\normtwo{
        Q\bigl(
            \ketbra{u}{t_j}\otimes\Id
        \bigr)P
    }^2
    +
    4\normtwo{
        Q\bigl(
            \ketbra{t_j}{u}\otimes\Id
        \bigr)P
    }^2.
    \label{eq:real-imaginary-parallelogram}
\end{align}
\else
\begin{align}
    \normtwo{
        Q\bigl(
            \partial_{e_j}S_d^\Phi\otimes\Id
        \bigr)P
    }^2
    +
    \normtwo{
        Q\bigl(
            \partial_{ie_j}S_d^\Phi\otimes\Id
        \bigr)P
    }^2=
    4\normtwo{
        Q\bigl(
            \ketbra{u}{t_j}\otimes\Id
        \bigr)P
    }^2
    +
    4\normtwo{
        Q\bigl(
            \ketbra{t_j}{u}\otimes\Id
        \bigr)P
    }^2.
    \label{eq:real-imaginary-parallelogram}
\end{align}
\fi
Since $0\leq Q\leq\Id$, we have
\begin{align*}
    \sum_{j=1}^{2^d-1}
    \normtwo{
        Q\bigl(
            \ketbra{u}{t_j}\otimes\Id
        \bigr)P
    }^2
    &=
    \sum_{j=1}^{2^d-1}
    \Tr\left(
        P
        \bigl(\ketbra{t_j}{u}\otimes\Id\bigr)
        Q
        \bigl(\ketbra{u}{t_j}\otimes\Id\bigr)
        P
    \right)
    \\
    \leq&
    \sum_{j=1}^{2^d-1}
    \Tr\left(
        P\bigl(\proj{t_j}\otimes\Id\bigr)P
    \right)
    =
    \Tr\left(
        P(\Pi_T\otimes\Id)P
    \right).
\end{align*}
Similarly, by invariance of the Hilbert--Schmidt norm under adjoint and
the same calculation with $P$ and $Q$ interchanged,
\[
    \sum_{j=1}^{2^d-1}
    \normtwo{
        Q\bigl(
            \ketbra{t_j}{u}\otimes\Id
        \bigr)P
    }^2
    =
    \sum_{j=1}^{2^d-1}
    \normtwo{
        P\bigl(
            \ketbra{u}{t_j}\otimes\Id
        \bigr)Q
    }^2
    \leq
    \Tr\left(
        Q(\Pi_T\otimes\Id)Q
    \right).
\]
Therefore, summing \cref{eq:real-imaginary-parallelogram} over $j$ gives, for each $x$,
\ifnum\fullpage=0
\begin{align*}
    &\sum_{j=1}^{2^d-1}
    \left(
        \normtwo{
            Q\bigl(
                \partial_{e_j}S_d^\Phi\otimes\Id
            \bigr)P
        }^2
        +
        \normtwo{
            Q\bigl(
                \partial_{ie_j}S_d^\Phi\otimes\Id
            \bigr)P
        }^2
    \right)
    \\
    &\qquad\leq
    4\Tr\left(
        P(\Pi_T\otimes\Id)P
    \right)
    +
    4\Tr\left(
        Q(\Pi_T\otimes\Id)Q
    \right).
\end{align*}
\else
\[
\sum_{j=1}^{2^d-1}
    \left(
        \normtwo{
            Q\bigl(
                \partial_{e_j}S_d^\Phi\otimes\Id
            \bigr)P
        }^2
        +
        \normtwo{
            Q\bigl(
                \partial_{ie_j}S_d^\Phi\otimes\Id
            \bigr)P
        }^2
    \right)
    \leq
    4\Tr\left(
        P(\Pi_T\otimes\Id)P
    \right)
    +
    4\Tr\left(
        Q(\Pi_T\otimes\Id)Q
    \right).
\]
\fi

For the phase direction,
\cref{lem:explicit-chfs-derivatives} gives $\bigl(
        \partial_{i\phi}S_d^\Phi
    \bigr)^\dagger
    \bigl(
        \partial_{i\phi}S_d^\Phi
    \bigr)
    =
    \Pi_E$ and
hence,
\[
    \normtwo{
        Q\bigl(
            \partial_{i\phi}S_d^\Phi\otimes\Id
        \bigr)P
    }^2
    \leq
    \Tr\left(
        P(\Pi_E\otimes\Id)P
    \right).
\]

Restoring the subscript $x$ and summing over $x$, the mutual
orthogonality of the spaces $E_x\oplus T_x$ gives
\[
\sum_{x\in\{0,1\}^d}
    \sum_{\delta}
    \normtwo{
        Q\bigl(
            \partial_{x,\delta}S_d^\Phi\otimes\Id
        \bigr)P
    }^2
    \leq
    5\Tr P+4\Tr Q
    \leq
    9r,
\]
where we used $\sum_x\Pi_{E_x}\leq\Id$,
$\sum_x\Pi_{T_x}\leq\Id$, and
$\Tr(R\Pi R)\leq\Tr R$ for orthogonal projections $R$ and $\Pi$. The final inequality is due to the rank condition.
\end{proof}

\begin{lemma}[One-sided CHFS derivative sum]
\label{lem:one-sided-chfs-sum}
Let $I$ be any register. For every vector $\xi$ on the CHFS query
register together with $I$,
\[
    \sum_{x\in\{0,1\}^d}
    \sum_{\delta}
    \norm{
        \bigl(
            \partial_{x,\delta}S_d^\Phi
            \otimes\Id_I
        \bigr)\xi
    }^2
    \leq
    5\cdot 2^d\norm\xi^2,
\]
where, for each $x$, $\delta$ ranges over the
basis $\{i\phi_x,
    e_{x,1},ie_{x,1},
    \ldots,
    e_{x,2^d-1},ie_{x,2^d-1}\}$.
\end{lemma}

\begin{proof}
We suppress the subscript $x$ in this proof; the distinct $x$'s are only
considered when we use the summation over $x\in\bit^d$.
By \cref{lem:explicit-chfs-derivatives}, we have
\[
    \bigl(
        \partial_{e_j}S_d^\Phi
    \bigr)^\dagger
    \bigl(
        \partial_{e_j}S_d^\Phi
    \bigr)
    +
    \bigl(
        \partial_{ie_j}S_d^\Phi
    \bigr)^\dagger
    \bigl(
        \partial_{ie_j}S_d^\Phi
    \bigr)
    =
    4\left(
        \proj u+\proj{t_j}
    \right).
\]
Also, for the phase direction, we have $
    \bigl(
        \partial_{i\phi}S_d^\Phi
    \bigr)^\dagger
    \bigl(
        \partial_{i\phi}S_d^\Phi
    \bigr)
    =
    \Pi_E.$
Therefore, summing over the tangent basis gives, for each $x$,
\ifnum\fullpage=0
\begin{align*}
    \sum_{\delta}
    \bigl(
        \partial_{\delta}S_d^\Phi
    \bigr)^\dagger
    \bigl(
        \partial_{\delta}S_d^\Phi
    \bigr)
    &=
    4(2^d-1)\proj u
    +
    4\Pi_T
    +
    \Pi_E
    \\
    &\leq
    (4\cdot 2^d-3)\Pi_E
    +
    4\Pi_T,
\end{align*}
\else
\[
    \sum_{\delta}
    \bigl(
        \partial_{\delta}S_d^\Phi
    \bigr)^\dagger
    \bigl(
        \partial_{\delta}S_d^\Phi
    \bigr)
    =
    4(2^d-1)\proj u
    +
    4\Pi_T
    +
    \Pi_E
    \leq
    (4\cdot 2^d-3)\Pi_E
    +
    4\Pi_T,
\]
\fi
where we used $\proj u\leq\Pi_E$.

Summing over all $x$, the mutual
orthogonality of the spaces $E_x\oplus T_x$ gives
\[
    \sum_{x\in\{0,1\}^d}
    \sum_{\delta}
    \bigl(
        \partial_{x,\delta}S_d^\Phi
    \bigr)^\dagger
    \bigl(
        \partial_{x,\delta}S_d^\Phi
    \bigr)
    \leq
    5\cdot 2^d\Id.
\]
Hence, expanding the left hand side becomes
\[
\sum_{x\in\{0,1\}^d}
    \sum_{\delta}
    \norm{
        \bigl(
            \partial_{x,\delta}S_d^\Phi
            \otimes\Id_I
        \bigr)\xi
    }^2
    =
    \left\langle
        \xi,
        \left(
            \sum_{x\in\{0,1\}^d}
            \sum_{\delta}
            \bigl(
                \partial_{x,\delta}S_d^\Phi
            \bigr)^\dagger
            \bigl(
                \partial_{x,\delta}S_d^\Phi
            \bigr)
            \otimes\Id_I
        \right)
        \xi
    \right\rangle
\]
which is bounded above by $5\cdot 2^d\norm\xi^2$, proving the lemma.
\end{proof}
\fi

\subsection{Oracle algorithm derivatives}
\label{subsec:channel-valued-chfs-functions}
This section extends the derivative to the oracle algorithms. We assume that in each query, algorithms fix the input length classically. Given the controlled access to the oracle, this can be done by, e.g., querying $S^\Phi_1,\ldots,S^\Phi_{d_{\max}}$ instead of a single coherent query for the maximum input size $d_{\max}$ at the cost of $d_{\max}$ query number blow up. See also \cite{TCC:DFH22}. In our main purpose, the algorithm runs in polynomial time so that $d_{\max}$ is at most polynomial.

\paragraph{Oracle algorithm as a function.}
Consider a PRU candidate algorithm $\mathcal A$ having oracle access to
$S^\Phi$. Let $A$ be the input/output register of $\mathcal A$, and let
$B$ be its ancillary register. After purifying intermediate measurements,
the algorithm can be written as
\begin{align}
W_\Phi
=
U_{q}S_{d_{q}}^\Phi
U_{q-1}
\cdots
U_1S_{d_1}^\Phi U_0,
\label{eq:oracle-algorithm-query-comb}
\end{align}
acting on the register $AB$, where $q$ is the total number
of oracle queries, $d_t$ is the classically fixed input length of the
$t$-th query, and every unitary $U_t$ is independent of $\Phi$.

We now fix a query input length $d$ and study the dependence of the algorithm only
on $\Phi_d$, while keeping all state families $\Phi_{d'}$ for $d'\neq d$ fixed. Let $q_d$ be the number of query positions $t$ satisfying $d_t=d$.
By absorbing all gates and queries at the other input lengths into the
inter-query unitaries and relabeling the length-$d$ query positions, we can write
\begin{align}
W_\Phi
=
U_{q_d}S_d^\Phi
U_{q_d-1}S_d^\Phi
\cdots
U_1S_d^\Phi U_0,
\label{eq:one-layer-query-comb}
\end{align}
where the unitaries $U_t$ may depend on the fixed state families
$\Phi_{d'}$ for $d'\neq d$, but are independent of $\Phi_d$.
Here we reuse the same $U$ for the query-irrelevant unitaries because they are less important in the analysis below.
For the analysis below, we also drop the subscript and write $\Phi$ and $S^{\Phi}$ to denote $\Phi_d$ and $S_d^{\Phi}$, respectively.

We consider the application of $\mathcal A$ to the maximally entangled
state $\ket{\Omega}_{AR}$, where $R\cong A$ is a reference register not
touched by the algorithm. The final states before and after tracing out
$B$ are
\begin{align}
\ket{\Xi_\Phi}_{ABR}
\defeq
(W_\Phi\otimes\Id_R)
\bigl(\ket\Omega_{AR}\ket0_B\bigr),
\qquad
\rho_\Phi
\defeq
\Tr_B\proj{\Xi_\Phi}.
\label{eq:reduced-choi-output}
\end{align}

We regard this application of the oracle algorithm as a function from the
product state space defining the length-$d$ CHFS oracle to the space of
Hermitian operators on $A\otimes R$. More formally, we consider the smooth
function
\[
F_d^{\mathcal A}:\Phi
\longmapsto
\rho_\Phi
\quad\text{from}\quad
\prod_{x\in\{0,1\}^d}\Sphere(\C^{2^d})
\quad\text{to}\quad
\Herm(A\otimes R),
\]
where $\Herm(A\otimes R)$ is regarded as a real Hilbert space with the
Hilbert-Schmidt inner product. We can therefore consider the derivatives
of $F_d^{\mathcal A}$ with respect to the states in $\Phi$.

\paragraph{Dealing with errors.}
Recall from \cref{eq:reduced-choi-output}.
Because of tracing out $B$, $\rho_\Phi$ may be mixed even though $\ket{\Xi_\Phi}$ is pure.
For each $\Phi$, define its distance from the set of
unitary Choi states by
\[
\eta(\Phi)
\defeq
\inf_{V\in\mathsf U(A)}
\normone{\rho_\Phi-
(V\otimes\Id_R)\ketbra\Omega(V\otimes\Id_R)^\dagger}.
\]
By compactness of $\mathsf U(A)$, the infimum must be attained by some $V$.
Choose any minimizer $V_\Phi$ and let
\[
P_\Phi
\defeq
(V_\Phi\otimes\Id_R)\ketbra\Omega
(V_\Phi\otimes\Id_R)^\dagger,
\qquad
\normone{\rho_\Phi-P_\Phi}=\eta(\Phi).
\]

Since $\ket{\Xi_\Phi}$ purifies $\rho_\Phi$,
\cref{lem:elementary-quantum-estimates} gives a pure state
$\ket{\chi_\Phi}_B$ such that the corresponding purification
of $P_\Phi$ approximates $\ket{\Xi_\Phi}$:
\[
    \ket{\Xi_\Phi}
    =
    \ket{\Xi_\Phi^{\mathrm{id}}}
    +
    \ket{e_\Phi},\quad
    \ket{\Xi_\Phi^{\mathrm{id}}}
    \defeq
    \bigl((V_\Phi\otimes\Id_R)\ket\Omega_{AR}\bigr)\ket{\chi_\Phi}_B,
    \quad
    \norm{e_\Phi}
    \leq
    \sqrt{\eta(\Phi)}.
\]

\paragraph{Derivatives.}
Fix $t\in[q_d]$. We occasionally consider the decomposition
\[
W_\Phi
=
W_{>t,\Phi}
S^\Phi
W_{<t,\Phi},
\]
where $W_{<t,\Phi}$ contains everything before the $t$-th query and
$W_{>t,\Phi}$ contains everything after it.
Applying the Leibniz rule of derivatives to \cref{eq:one-layer-query-comb} gives
\[
    \partial_{x,\delta}W_\Phi
    =
    \sum_{t=1}^{q_d}
    W_{>t,\Phi}
    \bigl(\partial_{x,\delta}S^\Phi\bigr)
    W_{<t,\Phi}.
\]
Define the $t$-th query part of the derivative of the
purified output by
\[
    \ket{\Delta_{t,x,\delta}}
    \defeq
    \left(
        W_{>t,\Phi}
        \bigl(\partial_{x,\delta}S^\Phi\bigr)
        W_{<t,\Phi}
        \otimes\Id_R
    \right)
    \bigl(\ket\Omega_{AR}\ket{0}_B\bigr).
\]
The Leibniz rule and linearity of the partial trace give
\[
    \partial_{x,\delta}\rho_\Phi
    =
    \sum_{t=1}^{q_d}
    \Tr_B\left(
        \ketbra{\Delta_{t,x,\delta}}{\Xi_\Phi}
        +
        \ketbra{\Xi_\Phi}{\Delta_{t,x,\delta}}
    \right).
\]

The squared
gradient norm of $F_d^{\mathcal A}$ at $\Phi$ is
\[
    \|\nabla F_d^{\mathcal A}(\Phi)\|_{\mathrm{HS}}^2
    :=
    \sum_{x\in\{0,1\}^d}
    \sum_{\delta}
    \normtwo{\partial_{x,\delta}\rho_\Phi}^2.
\]
The following lemma gives an upper bound on the gradient. 
\ifnum\submission=0
The remainder of this subsection is devoted to the proof of this lemma.

\begin{lemma}
\label{lem:gradient-energy-bound}
For every CHFS state family $\Phi$,
\[
 \|\nabla F_d^{\mathcal A}(\Phi)\|_{\mathrm{HS}}^2=
\sum_{x\in\{0,1\}^d}
\sum_{\delta}
\normtwo{
\partial_{x,\delta}\rho_\Phi
}^2
\leq
72q_d^2
\bigl(
1+2^d\eta(\Phi)
\bigr).
\]
\end{lemma}

\begin{proof}
By the preceding derivative identity and Cauchy--Schwarz,
\begin{align}
\normtwo{
\partial_{x,\delta}\rho_\Phi
}^2
\leq
q_d\sum_{t=1}^{q_d}
\normtwo{
\Tr_B\left(
\ketbra{\Delta_{t,x,\delta}}{\Xi_\Phi}
+
\ketbra{\Xi_\Phi}{\Delta_{t,x,\delta}}
\right)
}^2.
\label{eq:query-cauchy-schwarz}
\end{align}
Recall $\ket{\Xi_\Phi}
    =
    \ket{\Xi_\Phi^{\mathrm{id}}}
    +
    \ket{e_\Phi}$. To deal with them separately, define
\[
X_{t,x,\delta}
\defeq
\Tr_B\left(
\ketbra{\Delta_{t,x,\delta}}
{\Xi_\Phi^{\mathrm{id}}}
\right),
Y_{t,x,\delta}
\defeq
\Tr_B\left(
\ketbra{\Delta_{t,x,\delta}}{e_\Phi}
\right).
\]
Then, we have $\Tr_B\left(
\ketbra{\Delta_{t,x,\delta}}{\Xi_\Phi}
+
\ketbra{\Xi_\Phi}{\Delta_{t,x,\delta}}
\right)
=
X_{t,x,\delta}
+
X_{t,x,\delta}^\dagger
+
Y_{t,x,\delta}
+
Y_{t,x,\delta}^\dagger,$
and therefore
\begin{align}
\normtwo{
\Tr_B\left(
\ketbra{\Delta_{t,x,\delta}}{\Xi_\Phi}
+
\ketbra{\Xi_\Phi}{\Delta_{t,x,\delta}}
\right)
}^2
\leq
8\normtwo{X_{t,x,\delta}}^2
+
8\normtwo{Y_{t,x,\delta}}^2.
\label{eq:z-product-error-bound}
\end{align}

We first bound the terms $X_{t,x,\delta}$.
Let $\{\ket a\}_{a=1}^{d_A}$ be the orthonormal basis of $A$ used to
define $\ket\Omega_{AR}$. Define
\[
    \ket{s_a}
    \defeq
    W_{>t,\Phi}
    \bigl(\partial_{x,\delta}S^\Phi\bigr)
    W_{<t,\Phi}
    \bigl(\ket a_A\ket0_B\bigr).
\]
Then
\[
    \ket{\Delta_{t,x,\delta}}
    =
    \frac1{\sqrt{d_A}}
    \sum_{a=1}^{d_A}
    \ket{s_a}\ket a_R,\quad
    \ket{\Xi_\Phi^{\mathrm{id}}}
    =
    \frac1{\sqrt{d_A}}
    \sum_{b=1}^{d_A}
    V_\Phi\ket b_A\ket{\chi_\Phi}_B\ket b_R.
\]
Therefore,
\[
    X_{t,x,\delta}
    =
    \frac1{d_A}
    \sum_{a,b=1}^{d_A}
    (\Id_A\otimes\bra{\chi_\Phi})\ket{s_a}
    \bra{ b}V_\Phi^\dagger
    \otimes\ketbra {a} {b}_R.
\]
Since the elements $\{\ketbra {a} {b}_R\}_{a,b}$ can be regarded as orthonormal matrices with respect to the
Hilbert-Schmidt inner product,
\[
    \normtwo{X_{t,x,\delta}}^2
    =
    \frac1{d_A^2}
    \sum_{a,b=1}^{d_A}
    \norm{
        (\Id_A\otimes\bra{\chi_\Phi})\ket{s_a}
    }^2.
\]

Define the projections regarding the initial states and the final states and the operators $W_{<t}$ and $W_{>t}$ by
\ifnum\fullpage=0
\begin{align*}
P_{t,\Phi}
&\defeq
W_{<t,\Phi}
\bigl(\Id_A\otimes\proj{0}_B\bigr)
W_{<t,\Phi}^\dagger,
\\
Q_{t,\Phi}
&\defeq
W_{>t,\Phi}^\dagger
\bigl(
\Id_A\otimes
\ketbra{\chi_\Phi}{\chi_\Phi}_B
\bigr)
W_{>t,\Phi}.
\end{align*}
\else
\[
P_{t,\Phi}
\defeq
W_{<t,\Phi}
\bigl(\Id_A\otimes\proj{0}_B\bigr)
W_{<t,\Phi}^\dagger,\qquad
Q_{t,\Phi}
\defeq
W_{>t,\Phi}^\dagger
\bigl(
\Id_A\otimes
\ketbra{\chi_\Phi}{\chi_\Phi}_B
\bigr)
W_{>t,\Phi}.
\]
\fi
Both $P_{t,\Phi}$ and $Q_{t,\Phi}$ are orthogonal projections of rank
$d_A$.
Then we can write
\begin{align*}
\normtwo{X_{t,x,\delta}}^2
&=
\frac1{d_A^2}
\sum_{a,b}
\norm{
(\Id_A\otimes\bra{\chi_\Phi})\ket{s_a}
}^2
\\
&=
\frac1{d_A}
\sum_a
\norm{
(\Id_A\otimes\bra{\chi_\Phi})
W_{>t,\Phi}
\bigl(\partial_{x,\delta}S^\Phi\bigr)
W_{<t,\Phi}
\bigl(\ket a_A\ket0_B\bigr)
}^2
\\
&=
\frac1{d_A}
\Tr\left(
P_{t,\Phi}
\bigl(\partial_{x,\delta}S^\Phi\bigr)^\dagger
Q_{t,\Phi}
\bigl(\partial_{x,\delta}S^\Phi\bigr)
P_{t,\Phi}
\right)
\\
&=
\frac1{d_A}
\normtwo{
Q_{t,\Phi}
\bigl(\partial_{x,\delta}S^\Phi\bigr)
P_{t,\Phi}
}^2.
\end{align*}
Since $P_{t,\Phi}$ and $Q_{t,\Phi}$ have rank $d_A$,
\cref{lem:two-sided-chfs-sum} gives
\begin{align}
\sum_{x,\delta}
\normtwo{X_{t,x,\delta}}^2
=
\frac1{d_A}
\sum_{x,\delta}
\normtwo{
Q_{t,\Phi}
\bigl(\partial_{x,\delta}S^\Phi\bigr)
P_{t,\Phi}
}^2\leq
9.
\label{eq:product-derivative-total}
\end{align}

For the error term,
$\norm{e_\Phi}\leq\sqrt{\eta(\Phi)}$ gives
\[
\normtwo{Y_{t,x,\delta}}^2
=\normtwo{\Tr_B\left(
\ketbra{\Delta_{t,x,\delta}}{e_\Phi}
\right)}^2
\leq \norm{\Delta_{t,x,\delta}}^2 \norm{e_\Phi}^2
\leq
\eta(\Phi)
\norm{\Delta_{t,x,\delta}}^2
\]
where we use $\normtwo{\Tr_B(\ketbra{x}{y})} \le \norm{x}\norm{y}.$
The pre-query state
\[
(W_{<t,\Phi}\otimes\Id_R)
\bigl(\ket\Omega_{AR}\ket0_B\bigr)
\]
is a unit vector. Including the reference register and all other
non-query registers in the ancillary register of
\cref{lem:one-sided-chfs-sum}, and using invariance of the norm under
$W_{>t,\Phi}$, gives
\begin{align}
\sum_{x,\delta}
\normtwo{Y_{t,x,\delta}}^2
\leq
5\cdot 2^d\eta(\Phi).
\label{eq:error-derivative-total}
\end{align}

Summing \cref{eq:query-cauchy-schwarz} over $x,\delta$, and using
\cref{eq:z-product-error-bound,eq:product-derivative-total,eq:error-derivative-total},
gives
\[
8q_d^2
\bigl(
9+5\cdot 2^d\eta(\Phi)
\bigr)
\leq
72q_d^2
\bigl(
1+2^d\eta(\Phi)
\bigr),
\]
which proves the claim.
\end{proof}

\else
The proof of the following lemma is placed in \cref{app:chfsproofs}. The high level idea of the proof is by bounding the terms related to $\ket{\Xi_\Phi^{\mathrm{id}}}$ using \cref{lem:two-sided-chfs-sum}, and bounding the terms about $\ket{e_\Phi}$ using \cref{lem:one-sided-chfs-sum}.

\begin{lemma}
\label{lem:gradient-energy-bound}
For every CHFS state family $\Phi$,
\[
 \|\nabla F_d^{\mathcal A}(\Phi)\|_{\mathrm{HS}}^2=
\sum_{x\in\{0,1\}^d}
\sum_{\delta}
\normtwo{
\partial_{x,\delta}\rho_\Phi
}^2
\leq
72q_d^2
\bigl(
1+2^d\eta(\Phi)
\bigr).
\]
\end{lemma}
\fi

\subsection{CHFS oracle concentration}
\label{subsec:one-layer-resampling}
\label{subsec:simultaneous-large-layer-resampling}

This section presents the concentration phenomena for the output state $\rho_{\Phi}$ using the perspective of $F_d^{\mathcal A}.$
We prove the 1-norm bound instead of the 2-norm bound as the 1-norm is more convenient when discussing the distinguishability.

We first consider the CHFS states at a fixed input length $d$, while
keeping the state families at all other input lengths fixed.
Let $q_d$ be the number of length-$d$ queries made by the algorithm.
All the expectations are over the product Haar measure on the ambient space. In particular, for $\Phi=\Phi_d$, the Haar measure is over $\prod_{x\in\{0,1\}^d}\Sphere(\C^{2^d})$.

\begin{lemma}
\label{thm:one-layer-resampling}
Let $\Phi=\Phi_d$ be sampled from the product Haar measure
and define the expected output $\mu
\defeq
\Exp_\Phi\rho_\Phi$ and the expected error bound $\bar\eta
\defeq
\Exp_\Phi\eta(\Phi)$.
Then, 
for a universal constant $C>0$,
\[
\Exp_\Phi
\normone{\rho_\Phi-\mu}
\leq
Cq_d
\left(
2^{-d/2}
+
\sqrt{\bar\eta}
\right).
\]
\end{lemma}

\begin{proof}
If $q_d=0$, then $\rho_\Phi$ is independent of $\Phi$, and the claim is
immediate. Assume that $q_d\geq1$.

Let $\Phi'$ be an independent copy of $\Phi$, i.e., a Haar random element from the product space.
By independence,
\[
\Exp_{\Phi,\Phi'}
\normtwo{\rho_\Phi-\rho_{\Phi'}}^2
=
2\Exp_\Phi
\normtwo{\rho_\Phi-\mu}^2
\le 
\frac{2}{2^{d+1}-1}
\Exp_\Phi
\sum_{x,\delta}
\normtwo{
\partial_{x,\delta}\rho_\Phi
}^2,
\]
where we use the Poincar\'e inequality in
\cref{thm:hilbert-product-poincare}.
By \cref{lem:gradient-energy-bound},
\[
\Exp_{\Phi,\Phi'}
\normtwo{\rho_\Phi-\rho_{\Phi'}}^2
\leq
\frac{144q_d^2}{2^{d+1}-1}
\left(
1+2^d\bar\eta
\right)
\leq
144q_d^2
\left(
2^{-d}
+
\bar\eta
\right).
\]
Jensen's inequality gives
\[
\Exp_{\Phi,\Phi'}
\normtwo{\rho_\Phi-\rho_{\Phi'}}
\leq
12q_d
\left(
2^{-d/2}
+
\sqrt{\bar\eta}
\right).
\]

Now we turn this bound into the 1-norm bound.
Applying \cref{lem:rank-one-trace-norm-comparison} to
$\rho_\Phi,\rho_{\Phi'}$ and then averaging gives
\begin{align}
\Exp_{\Phi,\Phi'}
\normone{\rho_\Phi-\rho_{\Phi'}}
&\leq
(1+\sqrt2)
\Exp_{\Phi,\Phi'}
\bigl(
\eta(\Phi)+\eta(\Phi')
\bigr)
+
\sqrt2
\Exp_{\Phi,\Phi'}
\normtwo{\rho_\Phi-\rho_{\Phi'}}
\nonumber\\
&=
2(1+\sqrt2)\bar\eta
+
\sqrt2
\Exp_{\Phi,\Phi'}
\normtwo{\rho_\Phi-\rho_{\Phi'}}
\nonumber\\
&\leq
Cq_d
\left(
2^{-d/2}
+
\sqrt{\bar\eta}
\right).
\label{eq:one-layer-pairwise-trace-bound}
\end{align}
Here $0\leq\bar\eta\leq2$, so
$\bar\eta\leq\sqrt{2\bar\eta}$, and $q_d\geq1$.
Finally, for fixed $\Phi$, convexity of the trace norm gives
\[
\normone{\rho_\Phi-\mu}
=
\normone{
\Exp_{\Phi'}
\bigl(
\rho_\Phi-\rho_{\Phi'}
\bigr)
}
\leq
\Exp_{\Phi'}
\normone{\rho_\Phi-\rho_{\Phi'}}.
\]
Averaging over $\Phi$ and applying
\cref{eq:one-layer-pairwise-trace-bound} proves the claim.
\end{proof}

\begin{lemma}
\label{lem:simultaneous-large-layer-resampling}
Fix an integer $\tau\geq0$ and the CHFS state families
$\Phi_{\leq\tau}$. For each $d>\tau$, sample $\Phi_d$
independently from its product Haar measure, and let
$\Phi=(\Phi_{\leq\tau},\Phi_{>\tau})$ denote the full family
of CHFS states.
Let $\rho_\Phi$ and $\eta(\Phi)$ be the output state and its
distance from unitary Choi states defined above.
Let $q_{>\tau}\defeq\sum_{d>\tau}q_d$ and
$\mu\defeq\Exp_{\Phi_{>\tau}}\rho_\Phi$.
Then, for a universal constant $C>0$,
\[
\Exp_{\Phi_{>\tau}}\normone{\rho_\Phi-\mu}
\leq Cq_{>\tau}\left(
2^{-\tau/2}
+\sqrt{\Exp_{\Phi_{>\tau}}\eta(\Phi)}
\right).
\]
\end{lemma}

\begin{proof}
The claim is immediate if $q_{>\tau}=0$.
Otherwise, let $d_{\max}$ be the largest input length queried
by $\mathcal A$.
Let $\Phi'=(\Phi_{\leq\tau},\Phi'_{>\tau})$, where
$\Phi'_{>\tau}$ is an independent copy of $\Phi_{>\tau}$.

For $\tau\leq j\leq d_{\max}$, let $\Phi^{(j)}$ be the family
obtained from $\Phi$ by replacing $\Phi_d$ with $\Phi'_d$
for every $\tau<d\leq j$.
Thus $\rho_{\Phi^{(\tau)}}=\rho_\Phi$ and
$\rho_{\Phi^{(d_{\max})}}=\rho_{\Phi'}$, since $\mathcal A$
makes no queries of input length greater than $d_{\max}$.
Every $\Phi^{(j)}$ has the same distribution as $\Phi$,
with $\Phi_{\leq\tau}$ fixed.

For each $\tau<d\leq d_{\max}$, the families
$\Phi^{(d-1)}$ and $\Phi^{(d)}$ differ only at input length $d$.
Condition on their common state families and apply
\cref{eq:one-layer-pairwise-trace-bound} to the independent
pair $\Phi_d,\Phi'_d$.
Averaging over the common state families and applying
Jensen's inequality gives
\[
\Exp_{\Phi,\Phi'}
\normone{\rho_{\Phi^{(d-1)}}-\rho_{\Phi^{(d)}}}
\leq Cq_d\left(
2^{-d/2}
+\sqrt{\Exp_{\Phi_{>\tau}}\eta(\Phi)}
\right).
\]

Finally, convexity of the trace norm and the triangle
inequality give
\[
\Exp_{\Phi_{>\tau}}\normone{\rho_\Phi-\mu}
\leq
\Exp_{\Phi,\Phi'}\normone{\rho_\Phi-\rho_{\Phi'}}
\leq
\sum_{\tau<d\leq d_{\max}}
\Exp_{\Phi,\Phi'}
\normone{\rho_{\Phi^{(d-1)}}-\rho_{\Phi^{(d)}}}.
\]
Applying the preceding bound and using
$2^{-d/2}\leq2^{-\tau/2}$ and
$\sum_{\tau<d\leq d_{\max}}q_d=q_{>\tau}$
proves the claim.
\end{proof}

\subsection{CHFS oracle simulation}\label{subsec:CHFS-simulation}
Finally, we recall the simulation of the CHFS algorithm using phase-invariant designs \cite{AE07,Kre21,BEHMV26}.
A distribution $\mathcal D$ on the unit sphere of $\C^D$ is a
\emph{phase-invariant $\delta$-approximate $t$-design} if
\[
(1-\delta)\EE_{\psi\gets\Haar}(\proj\psi)^{\otimes s}
\preceq
\EE_{\psi\gets\mathcal D}(\proj\psi)^{\otimes s}
\preceq
(1+\delta)\EE_{\psi\gets\Haar}(\proj\psi)^{\otimes s},
\]
for $0\leq s\leq t$,
and, for $0\leq i,j\leq t$ with $i\neq j$, $
\EE_{\psi\gets\mathcal D}
\ket\psi^{\otimes i}\bra\psi^{\otimes j}=0.$
The zeroth tensor power is the scalar $1$.
The construction of such a phase-invariant approximate design is given in \cite{BEHMV26} using the standard design from \cite{BHH16}.

\begin{lemma}[Phase-invariant designs]
\label{lem:succinct-phase-invariant-designs}
There is a deterministic classical algorithm, in time $\poly(n,t,\log(1/\delta))$, that, given
$n,t\in\N$, $0<\delta<1$, and a seed $s$ of length
$\poly(n,t,\log(1/\delta))$, outputs a classical description
of an $n$-qubit unitary circuit $U_s$ of size
$\poly(n,t,\log(1/\delta))$.
For uniformly random $s$, the states
$\ket{\psi_s}=U_s\ket{0^n}$ form a phase-invariant
$\delta$-approximate $t$-design.
\end{lemma}

We use the following quantum-output version of \cite{BEHMV26}, which can be proven in essentially the same way.
\begin{lemma}[CHFS design simulation]
\label{lem:chfs-design-simulation}
Consider a $q$-query oracle algorithm $A$ given a CHFS oracle $S^\Phi$ for $\Phi = (\ket{\phi_x})_{x \in \bit^*}$ where the states in $\Phi$ are sampled independently from Haar measure.
Let $\Psi = \{\ket{\psi_x}\}_{x \in \bit^*}$ where each state $\ket{\psi_x}$ is independently sampled from a distribution of phase-invariant $\delta$-approximate $2q$-design states of the same length as $\ket{\phi_x}$.
Suppose that $A$ makes queries of input length at most $L$.
Let $\rho_{\Phi}$ and $\rho_{\Psi}$ be the output states of the algorithm before the measurement given oracle $S^\Phi$ and $S^\Psi$, respectively.
Then, there are universal constants $K>0$ and $C>1$ such that
\[
\normone{\EE_\Phi\rho_\Phi-\EE_\Psi\rho_\Psi}
\leq KqC^{q(L+1)}\delta.
\]
The same bound holds when the design seeds are $2q$-wise independent.
The output may include a reference register.
The same conclusion holds if the oracle states at any chosen
input lengths are fixed in both computations, with expectations
taken over the remaining states.
\end{lemma}
\begin{proof}[Proof sketch]
We follow the proof of \cite[Lemma~7]{BEHMV26} that gives
the advantage bound, keeping the query count $q$ and the
maximum input length $L$ separate.

In that proof, the advantage is bounded by
\[
\left|
\Pr_\Phi[A^{S^\Phi}\to1]
-\Pr_\Psi[A^{S^\Psi}\to1]
\right|
\leq O\left(
D^{2q}
\normone{
\EE_\Phi J(\mathcal S^\Phi)^{\otimes q}
-\EE_\Psi J(\mathcal S^\Psi)^{\otimes q}
}
\right),
\]
where $\mathcal S^\Phi,\mathcal S^\Psi$ are the channels
induced by the finite direct sums of the full controlled
CHFS gates at input lengths at most $L$, and $D$ is their
dimension.
This uses the argument of \cite[Lemmas~23 and~24]{Kre21},
which simulates the queries using Choi states at the cost
of a factor $D^{2q}$.

For the output states, append an arbitrary two-outcome
measurement to $A$ and apply the same bound.
Taking the supremum over measurements gives
\[
\normone{\EE_\Phi\rho_\Phi-\EE_\Psi\rho_\Psi}\leq O\left(
D^{2q}
\normone{
\EE_\Phi J(\mathcal S^\Phi)^{\otimes q}
-\EE_\Psi J(\mathcal S^\Psi)^{\otimes q}
}
\right),
\]
with the factor of two absorbed into the implicit constant.
Since $D=2^{O(L+1)}$ and the number of oracle labels is at most
$2^{L+1}$, the moment bound in that proof gives
$O(qC^{q(L+1)}\delta)$ for a universal constant $C>1$.
Each term in the moment expansion involves at most $2q$
distinct inputs, so $2q$-wise independent design seeds suffice.
\end{proof}
\section{Oracle Separation of PRUs from PRFSGs}
\label{sec:separation}

This section proves our main theorem.
\begin{theorem}
\label{thm:prfs-pru-separation}
There is a unitary oracle relative to which adaptively secure,
quantum-accessible PRFSGs exist, whereas PRUs do not.
\end{theorem}

Our oracle is $\mathcal O^\Phi\defeq(S^\Phi,\QPSPACE)$, where $S^\Phi$ is a
standard CHFS oracle. 
The construction of the PRFSGs is as follows: Let $\kappa,m=\omega(\log\lambda)$ be such that $\kappa+ m = \lambda$. The generation algorithm on the key $k \in \bit^\kappa$ and input $x \in \bit^m$ outputs the state $\ket{\phi_{k\| x}}$, which is obtained by applying $S_\lambda^{\Phi}$ on $\ket{k\|x,0}$. 
\cite{BEHMV26,EPRINT:GLMY25} proved that this construction is an adaptively secure quantum-accessible PRFSG relative to the oracle.

We focus on the non-existence of PRUs in the same oracle model.
Fix an oracle QPT algorithm $\mathcal G$ that is a candidate PRU construction, which is of the form \cref{eq:oracle-algorithm-query-comb} that acts on $n=n(\lambda)$ qubits for $n=\omega(\log \lambda)$ with a $\lambda$-bit key.
Let $q_d$ be the number of length-$d$ CHFS queries and let $q=\sum_d q_d$ be the total number of queries. Suppose that $\mathcal G$ runs in time $\lambda^a$ for a fixed constant $a$ that depends on $\mathcal G$. 
Let $d_{\max}\leq\lceil\lambda^a\rceil$ be the maximum query length.
We assume $q>0$, otherwise $\mathcal G$ is insecure due to $\QPSPACE.$\footnote{Alternatively, we assume that $q_d>0$ for all $d \le d_{\max}$ using a dummy query.}

\ifnum\submission=1
We first prove the dichotomy property of the target Choi state in \cref{subsec:Choi}, which presents the main object, the averaged Choi state, of the attack and shows that it satisfies a concentration bound or the construction is far from unitary. \cref{subsec:simulating-g} shows how to simulate the averaged Choi state without using the CHFS oracle. \cref{subsec:distingsuiher} presents the distinguisher and the proof of the main theorem.
\else
We first prove the dichotomy property of the target Choi state in \cref{subsec:Choi}, which presents the main object, the averaged Choi state, of the attack and shows that it satisfies a concentration bound or the construction is far from unitary. \cref{subsec:simulating-g} shows how to simulate the averaged Choi state without using the CHFS oracle. \cref{subsec:distingsuiher} and \cref{subsec:analysis} present the distinguisher and its analysis.
\fi

\subsection{Dichotomy of the averaged Choi state}\label{subsec:Choi}
The main ingredient of the distinguisher is the Choi state of the candidate algorithm $\mathcal G$. More precisely, we consider the average Choi state where the CHFS oracle for large inputs is replaced by an independent random CHFS oracle.

\paragraph{Choi states and independent Haar states.}
Let $\tau=\left\lceil(10a+19)\log_2\lambda\right\rceil$ be the threshold. 
Let $A$ be the $n$-qubit input register of $\mathcal G_k$, and let
$R$ be an $n$-qubit reference register. Recall that for a channel $\mathcal E$ on $A$
\[
\ket\Omega_{AR}=2^{-n/2}\sum_{u\in\{0,1\}^n}\ket u_A\ket u_R,
\qquad
J(\mathcal E)=(\mathcal E\otimes\mathcal I_R)(\proj\Omega).
\]
Thus, one copy of the Choi state $J(\mathcal E)$
is prepared by applying $\mathcal E$ to the $A$ register of
$\ket\Omega_{AR}$. Any ancillary registers of $\mathcal E$
are traced out. 
We are mainly interested in the Choi state
\[
 \rho_{k,\Phi}=J(\mathcal G_k^{\mathcal O^\Phi})
\]
that can be obtained by applying $\mathcal G$ once.

For families of pure states defining the CHFS oracles $\Phi$, we write $\Phi_{\leq\tau}=(\Phi_d)_{d\leq\tau}$ and
$\Phi_{>\tau}=(\Phi_d)_{d>\tau}$. 
We consider another family of states that are sampled from the identical distribution 
$\Phi'_{>\tau}=(\Phi'_d)_{d>\tau}$: for each $d>\tau$ and $x\in\{0,1\}^d$, independently
sample a $d$-qubit Haar state $\ket{\phi'_{d,x}}$ and use these states
to define $S_d^{\Phi'}$ by the same CHFS construction. These states are independent of the original oracle $\Phi$.

The state $\rho_{k,(\Phi_{\leq\tau},\Phi'_{>\tau})}$ denotes the
Choi state for $\mathcal G$ using the CHFS oracle with $\Phi_{\leq\tau},\Phi'_{>\tau}$, i.e., one with the original $S_d^{\Phi}$ for 
$d\leq\tau$ and $S_d^{\Phi'}$ for $d>\tau$. 
Let
\[
 \mu_k\defeq
 \EE_{\Phi'_{>\tau}}
 \left[\rho_{k,(\Phi_{\leq\tau},\Phi'_{>\tau})}\right]
\]
be the average of the states $\rho_{k,(\Phi_{\leq\tau},\Phi'_{>\tau})}$ over (product) Haar random $\Phi'_{>\tau}$.
In particular, $\mu_k$ depends on the original oracle only through
$\Phi_{\leq\tau}$.

\paragraph{Dichotomy.}
The next lemma shows that either the Choi states have noticeable average mixedness over the key, or they are close on average to the states $\mu_k$.

\begin{lemma}[Mixedness or concentration]
\label{lem:purity-concentration-dichotomy}
With probability one over $\Phi$, for all sufficiently large
$\lambda$, either
\[
\EE_k[1-\Tr(\rho_{k,\Phi}^2)]\geq\lambda^{-(4a+9)} \quad \text{ or }\quad\EE_k\normone{\rho_{k,\Phi}-\mu_k}\leq\lambda^{-1}.
\]
\end{lemma}

\begin{proof}
Fix $\Phi_{\leq\tau}$.
If $d_{\max}\leq\tau$, then $\rho_{k,\Phi}=\mu_k$ and the claim
is immediate. Otherwise, consider two choices
$\Phi_{>\tau}$ and $\Phi'_{>\tau}$.

Observe that
for unit vectors $\phi,\phi'$, the definition in
\cref{eq:one-label-chfs-unitary} gives
\[
\normop{S_\phi-S_{\phi'}}
\leq
\normop{\proj{\phi,1}-\proj{\phi',1}}
+2\norm{\phi-\phi'}
\leq4\norm{\phi-\phi'}.
\]
Since $S_d^\Phi$ is block diagonal in the input $x$,
\[
\normop{S_d^\Phi-S_d^{\Phi'}}
=
\max_{x\in\{0,1\}^d}
\normop{S_{\phi_{d,x}}-S_{\phi'_{d,x}}}
\leq
4\max_{x\in\{0,1\}^d}
\norm{\phi_{d,x}-\phi'_{d,x}}.
\]
The same operator-norm bound holds for the controlled gates.
The diamond distance between the channels induced by two unitaries
is at most twice their operator-norm distance. Thus a hybrid
argument over the queries gives
\begin{align*}
\normone{
\rho_{k,\Phi}
-\rho_{k,(\Phi_{\leq\tau},\Phi'_{>\tau})}
}
&\leq
2\sum_{\tau<d\leq d_{\max}}
q_d\normop{S_d^\Phi-S_d^{\Phi'}}\\
&\leq
8\sum_{\tau<d\leq d_{\max}}
q_d\max_{x\in\{0,1\}^d}
\norm{\phi_{d,x}-\phi'_{d,x}}\\
&\leq
8q\left(
\sum_{\tau<d\leq d_{\max}}
\sum_{x\in\{0,1\}^d}
\norm{\phi_{d,x}-\phi'_{d,x}}^2
\right)^{1/2}.
\end{align*}
Since $|\Tr(\rho^2)-\Tr(\sigma^2)|\leq2\normone{\rho-\sigma}$,
the function
$\Phi_{>\tau}\mapsto\EE_k[1-\Tr(\rho_{k,\Phi}^2)]$
is $16q$-Lipschitz in the product Euclidean distance.

For each $(d,x)$ with $\tau<d\leq d_{\max}$, write
$\phi_{d,x}=U_{d,x}\ket{0^d}$ with independent Haar-random
unitaries $U_{d,x}$. Since
$\norm{U\ket0-V\ket0}\leq\normtwo{U-V}$, this Lipschitz
bound also holds in the product Frobenius norm.
We can therefore apply \cite[Theorem~10]{Kre21}, with minimum
unitary dimension at least $2^{\tau+1}$.

If
$\EE_{k,\Phi_{>\tau}}[1-\Tr(\rho_{k,\Phi}^2)]\geq2\lambda^{-(4a+9)}$,
the lower-tail bound in \cite[Theorem~10]{Kre21} gives
\[
\Prb_{\Phi_{>\tau}}
\left[\EE_k[1-\Tr(\rho_{k,\Phi}^2)]<\lambda^{-(4a+9)}\right]
\leq
\exp\left(-\frac{(2^{\tau+1}-2)\lambda^{-2(4a+9)}}{24(16q)^2}\right)
=\exp(-\Omega(\lambda)).
\]

If instead
$\EE_{k,\Phi_{>\tau}}[1-\Tr(\rho_{k,\Phi}^2)]<2\lambda^{-(4a+9)}$,
\cref{lem:purity-to-unitary-choi} and Jensen's inequality give
\[
\EE_{k,\Phi_{>\tau}}e(\mathcal G_k^{\mathcal O^\Phi})
\leq5\sqrt{
\EE_{k,\Phi_{>\tau}}[1-\Tr(\rho_{k,\Phi}^2)]
}
<5\sqrt{2\lambda^{-(4a+9)}}.
\]
Applying the CHFS oracle concentration lemma (\cref{lem:simultaneous-large-layer-resampling}) for each key
and then Jensen's inequality yields
\ifnum\fullpage=0
\begin{align*}
\EE_{k,\Phi_{>\tau}}\normone{\rho_{k,\Phi}-\mu_k}
&\leq Cq\left(
2^{-\tau/2}
+\sqrt{\EE_{k,\Phi_{>\tau}}e(\mathcal G_k^{\mathcal O^\Phi})}
\right)\\
&\leq C_1q\left(2^{-\tau/2}+\lambda^{-(4a+9)/4}\right).
\end{align*}
\else
\[
\EE_{k,\Phi_{>\tau}}\normone{\rho_{k,\Phi}-\mu_k}
\leq Cq\left(
2^{-\tau/2}
+\sqrt{\EE_{k,\Phi_{>\tau}}e(\mathcal G_k^{\mathcal O^\Phi})}
\right)\leq C_1q\left(2^{-\tau/2}+\lambda^{-(4a+9)/4}\right).
\]
\fi
Markov's inequality therefore gives
\[
\Prb_{\Phi_{>\tau}}\left[
\EE_k\normone{\rho_{k,\Phi}-\mu_k}>\lambda^{-1}
\right]
\leq C_1\lambda q\left(2^{-\tau/2}+\lambda^{-(4a+9)/4}\right)
=O(\lambda^{-5/4}).
\]
In either case, the probability that both alternatives fail is
$O(\lambda^{-5/4})$, uniformly over $\Phi_{\leq\tau}$.
Averaging over $\Phi_{\leq\tau}$ gives the same bound over $\Phi$.
These failure probabilities are summable in $\lambda$, so the
Borel--Cantelli lemma proves the claim.
\end{proof}

\subsection{Simulating \texorpdfstring{$\mathcal G$}{G}}
\label{subsec:simulating-g}
The distinguisher uses the approximation of $\mu_k$. 
We explain that, after preprocessing on $\Phi_{\le \tau}$ with tomography, the approximation of $\mu_k$ can be efficiently prepared for each $k$ without querying $S^\Phi$ at all.

\paragraph{Tomography of the small layers.}
For every $d\leq\tau=\left\lceil(10a+19)\log_2\lambda\right\rceil$,
the distinguisher performs process tomography in \cref{thm:process-tomography} on the controlled gate
$\ketbra{0}_C\otimes\Id+\ketbra{1}_C\otimes S_d^\Phi$ with diamond-norm
accuracy $1/400q$ and failure probability $2^{-2\lambda}$. Let
$\widetilde S_d$ be the unitary specified by the classical output of
tomography, and let
$\widetilde S_{\leq\tau}=(\widetilde S_d)_{d\leq\tau}$ denote the collection of these descriptions. 

Each reconstructed gate acts on $2d+2=O(\log\lambda)$ qubits. By
\cref{thm:process-tomography}, the procedure runs in polynomial time
and uses polynomially many queries to $S^\Phi$. A union bound shows
that all reconstructed gates have the stated accuracy except with
probability at most $(\tau+1)2^{-2\lambda}$.

\paragraph{Preparing the simulated Choi state.}
Given $\widetilde S_{\leq\tau}$ and $\Phi'_{>\tau}$, let
$\widetilde{\mathcal G}_k^{\Phi'_{>\tau}}$ be the channel obtained
from $\mathcal G_k$ by replacing each controlled CHFS query of length
$d\leq\tau$ by $\widetilde S_d$, and each query of length $d>\tau$
by the corresponding controlled gate for $S_d^{\Phi'}$.
All other operations and queries to $\QPSPACE$ are unchanged.

We approximate the average state
$\EE_{\Phi'_{>\tau}}J(\widetilde{\mathcal G}_k^{\Phi'_{>\tau}})$
using phase-invariant designs, following the simulation in
\cite{BEHMV26}.
Let $\mathcal T$ be the deterministic classical algorithm of
\cref{lem:succinct-phase-invariant-designs}. More precisely,
on input $(d,2q,\delta;s)$, where $s$ is a design seed,
$\mathcal T$ outputs a classical description of a $d$-qubit
unitary circuit $U_{d,s}$ of size $\poly(d,q,\log(1/\delta))$.
For uniform $s$, the states $U_{d,s}\ket{0^d}$ form a
phase-invariant $\delta$-approximate $2q$-design.
Choose $\delta$ according to \cref{lem:chfs-design-simulation}
so that \cref{eqn:deltacondition} holds. Note that $\log(1/\delta) = O(\lambda^{2a})$.

Let $\{s_z\}_{z\in\bit^\ell}$ be an efficiently computable family of $2q$-wise independent functions on the inputs $(d,x)$ with $\tau<d\leq d_{\max}$ and $x\in\bit^d$ with $\ell=\poly(\lambda)$.
The output length of $s_z(d,x)$ equals the seed length for $\mathcal T(d,2q,\delta;\cdot)$.
For each $\tau<d\leq d_{\max}$ and $x\in\bit^d$,
let $\ket{\psi_{d,x,z}}$
be the state $U_{d,s_z(d,x)}\ket{0^d}$, where
$\mathcal T(d,2q,\delta;s_z(d,x))$ outputs the description
of $U_{d,s_z(d,x)}$.
We consider, for each $z$, the simulation of $\widetilde{\mathcal G}_k^{\Phi'_{>\tau}}$ by replacing the state $\ket{\phi'_{d,x}}$ for the CHFS oracle by $\ket{\psi_{d,x,z}}$.
Let $\widehat\mu_k$ be the Choi state of
this simulation, averaged over uniform $z$.

\emph{Circuit for $\widehat\mu_k$.}
Given $\lambda,k$ and $\widetilde S_{\leq\tau}$, the circuit
prepares a purification of $\widehat\mu_k$ as follows.
\begin{enumerate}
\item Prepare an $\ell$-qubit register $Z$ in
$H^{\otimes\ell}\ket{0^\ell}$ and prepare $\ket\Omega_{AR}$.

\item Run the purified computation of $\mathcal G_k$ on $A$.
For each CHFS query with $d\leq\tau$, apply $\widetilde S_d$.
For each CHFS query with $d>\tau$ and input $x$, perform the
following operations controlled on $Z$:
\begin{enumerate}
\item Compute the design seed $s_z(d,x)$.
\item Run $\mathcal T(d,2q,\delta;s_z(d,x))$
to compute the circuit description of $U_{d,s_z(d,x)}$.
\item Use this circuit and its inverse to implement the controlled
CHFS gate $S_{\psi_{d,x,z}}$ on the query target register,
with the original query control.
\item Uncompute the circuit description and the design seed.
\end{enumerate}
The seed and circuit description are computed reversibly.
All other gates and $\QPSPACE$ queries are unchanged. 

\item Output all registers.
\end{enumerate}
For fixed $z$, this is the simulation defining $\widehat\mu_k$.
Tracing out $Z$ and all ancillary registers gives its Choi state
averaged over uniform $z$, namely $\widehat\mu_k$.
Thus the output is a purification of $\widehat\mu_k$.

For the chosen $\delta$, both $s_z(d,x)$ and
$\mathcal T(d,2q,\delta;s_z(d,x))$ can be computed in polynomial
time. Since $\mathcal T$ outputs circuits of polynomial size,
the corresponding unitaries, their inverses, and the controlled
CHFS gates can also be implemented in polynomial time.
Thus the preparation procedure is a QPT algorithm with
access to $\QPSPACE$.

\ifnum\submission=1
The following lemma ensures the accuracy of the above circuit. The proof, which can be found in \cref{app:proofsdisting}, uses the simulation lemma in \cref{lem:chfs-design-simulation} and the results from the process tomography in \cref{thm:process-tomography}.

\begin{lemma}
\label{lem:average-choi-preparation}
Given $\lambda,k$ and the descriptions
$\widetilde S_{\leq\tau}$, the above procedure prepares a purification of $\widehat\mu_k$.
Except with probability at most $(\tau+1)2^{-2\lambda}$ over
tomography, for every key $k$,
$\normone{\widehat\mu_k-\mu_k}\leq1/200$.
\end{lemma}

\else
The following lemma ensures the accuracy of the above circuit.

\begin{lemma}
\label{lem:average-choi-preparation}
Given $\lambda,k$ and the descriptions
$\widetilde S_{\leq\tau}$, the above procedure prepares a purification of $\widehat\mu_k$.
Except with probability at most $(\tau+1)2^{-2\lambda}$ over
tomography, for every key $k$,
$\normone{\widehat\mu_k-\mu_k}\leq1/200$.
\end{lemma}

\begin{proof}
The preparation claim follows from the construction above.
Fix the tomography descriptions, which are used for all keys. Apply
\cref{lem:chfs-design-simulation} to
$\widetilde{\mathcal G}_k$ acting on one half of $\ket\Omega$ for the CHFS oracles with input length $d>\tau$.
The computation makes at most $q$ CHFS oracle queries, each of
length at most $\lceil\lambda^a\rceil$. Thus
\begin{align}
\label{eqn:deltacondition}
\normone{
\widehat\mu_k-
\EE_{\Phi'_{>\tau}}J(\widetilde{\mathcal G}_k^{\Phi'_{>\tau}})
}
\leq KqC^{q(\lceil\lambda^a\rceil+1)}\delta
\leq\frac1{400},
\end{align}
where $K,C$ are the universal constants in that lemma.
Since $q\leq\lambda^a$, this choice requires only
$\log(1/\delta)=O(\lambda^{2a})$.

Now suppose that all tomography estimates have error at most $1/400q$. For every key $k$ and every $\Phi'_{>\tau}$, a hybrid
argument over the queries of size $\le \tau$ gives
\[
\normone{
J(\widetilde{\mathcal G}_k^{\Phi'_{>\tau}})
-\rho_{k,(\Phi_{\leq\tau},\Phi'_{>\tau})}
}
\leq\sum_{d\leq\tau}\frac{q_d}{400q}\leq\frac1{400}.
\]
Averaging over $\Phi'_{>\tau}$ and combining
the two bounds gives
$\normone{\widehat\mu_k-\mu_k}\leq1/200$.
A union bound over all the process tomography procedures gives failure probability
at most $(\tau+1)2^{-2\lambda}$.
\end{proof}
\fi

\subsection{The distinguisher}\label{subsec:distingsuiher}
Now we are ready to describe the distinguisher. 
The distinguisher is given oracle access to a channel $\mathcal V$,
where either
$\mathcal V=\mathcal G_{k^\star}^{\mathcal O^\Phi}$ for a uniformly
random key $k^\star\in\{0,1\}^\lambda$, or $\mathcal V$ is an
$n$-qubit Haar-random unitary channel.

Let $r=\lceil143\lambda\rceil$. 
The distinguisher uses the following binary measurement $M_k$ for each $k$ that is defined as follows. 
\begin{description}
\item[Input:] $r$ copies of $J(\mathcal V)$.
\item[Preparation:] Prepare $r$ independent copies of $\widehat\mu_k$,
using fresh design seeds for each copy and the same tomography
descriptions $\widetilde S_{\leq\tau}$ for all copies.
\item[Measurement:] Apply a swap test to each pair and accept if
at least $2r/3$ tests pass.
\end{description}
It is clearly a QPT algorithm with access to $\QPSPACE$ given $k$ and $\widetilde S_{\leq\tau}$, given that the preparation of $\widehat\mu_k$ is done in QPT with access to $\QPSPACE.$

Our distinguisher is defined as follows.

\paragraph{Distinguisher
$\mathcal D^{\mathcal V,\mathcal O^\Phi}$.}
\begin{enumerate}
\item Run the process tomography procedure of \cref{subsec:simulating-g} to
obtain $\widetilde S_{\leq\tau}$.
\item Prepare $r+2$ copies of $\ket\Omega_{AR}$ and apply
$\mathcal V$ to obtain $J(\mathcal V)^{\otimes (r+2)}.$
\item Apply the purity test of \cref{lem:swap-test} to the first two
copies of $J(\mathcal V)$. If it fails, output $1$ and abort.
\item Otherwise, apply the quantum OR test to
$\{M_k:k\in\{0,1\}^\lambda\}$ on the remaining $r$ copies of
$J(\mathcal V)$. Output $1$ if it accepts and $0$ otherwise.
\end{enumerate}

\ifnum\submission=1
We prove the following lemmas in \cref{app:proofsdisting}.

\begin{lemma}[QPT implementation]\label{lem:QPTimplement}
The distinguisher is a uniform QPT algorithm with oracle access to
$\mathcal V$ and $\mathcal O^\Phi=(S^\Phi,\QPSPACE)$.
\end{lemma}

\begin{lemma}[Haar case]
\label{lem:haar-case}
If $\mathcal V$ is a Haar-random unitary, the distinguisher
outputs $1$ with negligible probability.
\end{lemma}

\begin{lemma}[Candidate case]
\label{lem:candidate-case}
Fix an oracle $\Phi$ for which the conclusion of
\cref{lem:purity-concentration-dichotomy} holds. If
$\mathcal V=\mathcal G_{k^\star}^{\mathcal O^\Phi}$ for a uniformly
random $k^\star$, then the distinguisher outputs 1 with non-negligible probability. More precisely,
$\Prb[\mathcal D^{\mathcal V,\mathcal O^\Phi}=1]
\geq\frac12\lambda^{-(4a+9)}-\negl(\lambda)
$ for all sufficiently large $\lambda$.
\end{lemma}

We are now ready to prove our main theorem.

\begin{proof}[Proof of \cref{thm:prfs-pru-separation}]
Enumerate the candidate PRUs $(\mathcal G_i)_{i \in \mathbb N}$, which form a countable set because uniform oracle machines with explicit polynomial time bounds form a countable set.
For each fixed $\mathcal G_i$, the dichotomy in \cref{lem:purity-concentration-dichotomy} holds with probability one over $\Phi$. 
By countable subadditivity,
with probability one over the CHFS oracle,
the dichotomy holds for all candidates and the PRFSG security (which happens with probability 1 \cite{BEHMV26}) holds for the construction described in the introduction of this section. Fix $\Phi$ so that this holds.

It suffices to discuss the non-existence of secure PRUs. 
For every candidate $\mathcal G$, the oracle QPT distinguisher $\mathcal D$ distinguishes the candidate
from a Haar-random unitary channel with inverse-polynomial advantage by \cref{lem:haar-case,lem:candidate-case}.
Hence PRUs do not exist relative to $\mathcal O^\Phi$.
\end{proof}

\else

\begin{lemma}[QPT implementation]
The distinguisher is a uniform QPT algorithm with oracle access to
$\mathcal V$ and $\mathcal O^\Phi=(S^\Phi,\QPSPACE)$.
\end{lemma}
\begin{proof}
The tomography procedure runs in polynomial time as shown in \cref{subsec:simulating-g}.
For fixed tomography descriptions, the preparation for $\widehat\mu_k$ and the $r=O(\lambda)$ swap tests give unitary QPT implementations of the measurements $M_k$ with access to $\QPSPACE$, whose circuits can be generated uniformly from $k$ and $\widetilde S_{\leq\tau}$.
There are $2^\lambda$ measurements, so \cref{cor:qpspace-quantum-or} gives a QPT implementation
of the quantum OR test with access to $\QPSPACE$.
All remaining operations have polynomial size.
\end{proof}

We prove the following lemmas in the next subsection.

\begin{lemma}[Haar case]
\label{lem:haar-case}
If $\mathcal V$ is a Haar-random unitary, the distinguisher
outputs $1$ with negligible probability.
\end{lemma}

\begin{lemma}[Candidate case]
\label{lem:candidate-case}
Fix an oracle $\Phi$ for which the conclusion of
\cref{lem:purity-concentration-dichotomy} holds. If
$\mathcal V=\mathcal G_{k^\star}^{\mathcal O^\Phi}$ for a uniformly
random $k^\star$, then the distinguisher outputs 1 with non-negligible probability. More precisely,
$\Prb[\mathcal D^{\mathcal V,\mathcal O^\Phi}=1]
\geq\frac12\lambda^{-(4a+9)}-\negl(\lambda)
$ for all sufficiently large $\lambda$.
\end{lemma}

We are now ready to prove our main theorem.

\begin{proof}[Proof of \cref{thm:prfs-pru-separation}]
Enumerate the candidate PRUs $(\mathcal G_i)_{i \in \mathbb N}$, which form a countable set because uniform oracle machines with explicit polynomial time bounds form a countable set.
For each fixed $\mathcal G_i$, the dichotomy in \cref{lem:purity-concentration-dichotomy} holds with probability one over $\Phi$. 
By countable subadditivity,
with probability one over the CHFS oracle,
the dichotomy holds for all candidates and the PRFSG security (which happens with probability 1 \cite{BEHMV26}) holds for the construction described in the introduction of this section. Fix $\Phi$ so that this holds.

It suffices to discuss the non-existence of secure PRUs. 
For every candidate $\mathcal G$, the oracle QPT distinguisher $\mathcal D$ distinguishes the candidate
from a Haar-random unitary channel with inverse-polynomial advantage by \cref{lem:haar-case,lem:candidate-case}.
Hence PRUs do not exist relative to $\mathcal O^\Phi$.
\end{proof}

\subsection{Analysis}\label{subsec:analysis}

\begin{proof}[Proof of \cref{lem:haar-case}]
Let $\mathcal V=\mathcal U$ for a Haar-random unitary $U$.
The purity test always passes because $J(\mathcal U)$ is pure.
It therefore suffices to bound the acceptance probability of
the quantum OR test.

Fix the original oracle and the tomography outcomes.
Every $\widehat\mu_k$ is then fixed independently of $U$.
For each key $k$, a swap test on $J(\mathcal U)$ and
$\widehat\mu_k$ passes with probability
$(1+\Tr(J(\mathcal U)\widehat\mu_k))/2$.
Haar invariance gives
\[
\EE_U J(\mathcal U)=\frac{\Id_{AR}}{2^{2n}},
\qquad
\EE_U\left[\frac{1+\Tr(J(\mathcal U)\widehat\mu_k)}2\right]
=\frac12+\frac1{2^{2n+1}}.
\]
This swap test uses one query to $U$. Hence
\cref{cor:haar-concentration} gives, for each $k$,
\[
\Prb_U\left[
\frac{1+\Tr(J(\mathcal U)\widehat\mu_k)}2>\frac7{12}
\right]
\leq\exp(-\Omega(2^n)).
\]
A union bound over the keys shows that, except with probability
at most $2^\lambda\exp(-\Omega(2^n))=\negl(\lambda)$,
the swap-test acceptance probability is at most $7/12$
for every key. We only consider this case. Then, by Chernoff inequality 
(\cref{lem:bernoulli-chernoff}) and $r=\lceil143\lambda\rceil$,
each $M_k$ accepts with probability at most $2^{-3\lambda}$.
Thus the quantum OR test accepts and the distinguisher outputs 1
with probability at most $4\cdot2^\lambda\cdot2^{-3\lambda}$ by \cref{lem:quantum-or}.
Overall, the distinguisher outputs $1$
with negligible probability.
\end{proof}

\begin{proof}[Proof of \cref{lem:candidate-case}]
We first consider the mixedness case, i.e., it holds that $\EE_k[1-\Tr(\rho_{k,\Phi}^2)]\geq\lambda^{-(4a+9)}$. In this case,
the purity test alone outputs $1$ with probability at least
$\frac12\lambda^{-(4a+9)}$.

Otherwise, \cref{lem:purity-concentration-dichotomy} gives
$\EE_k\normone{\rho_{k,\Phi}-\mu_k}\leq\lambda^{-1}$.
By Markov's inequality, all but a
$100\lambda^{-(4a+9)}+100\lambda^{-1}$ fraction of keys satisfy the near-purity and closeness conditions simultaneously:
\[
1-\Tr(\rho_{k,\Phi}^2)\leq\frac1{100},
\qquad
\normone{\rho_{k,\Phi}-\mu_k}\leq\frac1{100}.
\]
Condition on successful tomography. 
By \cref{lem:average-choi-preparation}, every such key satisfies
$\normone{\rho_{k,\Phi}-\widehat\mu_k}\leq3/200$, and consequently
\[
\Tr(\rho_{k,\Phi}\widehat\mu_k)
\geq\Tr(\rho_{k,\Phi}^2)
-\normone{\rho_{k,\Phi}-\widehat\mu_k}
\geq\frac{39}{40}.
\]
For each such key $k$, every swap test in $M_k$ passes
with probability at least $79/80$.
Since $M_k$ rejects only if fewer than $2r/3$ of its $r$
independent swap tests pass, the Chernoff-Hoeffding bound
(\cref{lem:bernoulli-chernoff}) gives a rejection probability
of at most $2^{-3\lambda}$.
For the actual key $k^\star$ satisfying both the near-purity and closeness conditions,
the quantum OR test (\cref{lem:quantum-or}) thus accepts with probability at least
$(1-2^{-3\lambda})^2/7$.

The purity test acts on separate copies of $J(\mathcal V)$, and its
failure also causes the distinguisher to output $1$. Therefore,
averaging over the key and accounting for tomography failure gives
\[
\Prb[\mathcal D^{\mathcal V,\mathcal O^\Phi}=1]
\geq
\left(1-100\lambda^{-(4a+9)}-\frac{100}{\lambda}\right)
\frac{(1-2^{-3\lambda})^2}{7}
-(\tau+1)2^{-2\lambda}
\]
where $(\tau+1)2^{-2\lambda}$ is due to the tomography failure probability from \cref{lem:average-choi-preparation}.
This is bounded below by a positive constant for all sufficiently
large $\lambda$, proving the claim.
\end{proof}

\fi

\bibliographystyle{plain}
\bibliography{cryptobib/abbrev3,cryptobib/crypto,ref}

\ifnum\submission=1
\appendix
\crefalias{section}{appendix}
\section{Missing Proofs}
\subsection{Proofs for preliminaries}\label{app:prelproofs}

We restate the lemmas and provide the proofs.
\begin{lemma}[\cref{lem:elementary-quantum-estimates}]
Let $\ket\Xi_{AB}$ be a purification of $\rho_A$.
If $\normone{\rho-\proj\psi}\leq\eps\leq2$ for a pure state
$\ket\psi_A$, then some unit vector $\ket\eta_B$ satisfies
$\norm{\ket\Xi-\ket\psi\ket\eta}\leq\sqrt\eps$.
\end{lemma}

\begin{proof}[Proof of \cref{lem:elementary-quantum-estimates}]
Let $w=(\bra\psi\otimes\Id_B)\ket\Xi$ and
$p=\norm w^2=\Tr(\proj\psi\rho)$.
The measurement $\{\proj\psi,\Id-\proj\psi\}$ accepts
$\proj\psi$ with probability one and $\rho$ with probability $p$.
Hence $1-p\leq\normone{\rho-\proj\psi}/2\leq\eps/2$.
If $p>0$, let $\ket\eta=w/\sqrt p$; then
$\norm{\ket\Xi-\ket\psi\ket\eta}^2=2-2\sqrt p\leq\eps$.
If $p=0$, then $\eps=2$ and any unit vector $\ket\eta$ suffices.
\end{proof}

\begin{lemma}[\cref{lem:rank-one-trace-norm-comparison}]
Let $\rho,\sigma$ be density operators and $P,Q$ rank-one projectors
on the same space. If $\epsilon_\rho=\normone{\rho-P}$ and
$\epsilon_\sigma=\normone{\sigma-Q}$, then $
\normone{\rho-\sigma}
\leq(1+\sqrt2)(\epsilon_\rho+\epsilon_\sigma)
+\sqrt2\normtwo{\rho-\sigma}.$
\end{lemma}

\begin{proof}[Proof of \cref{lem:rank-one-trace-norm-comparison}]
The triangle inequality and
$\normone{P-Q}=\sqrt2\normtwo{P-Q}$ give
\[
\normone{\rho-\sigma}
\leq\epsilon_\rho+\epsilon_\sigma+\sqrt2\normtwo{P-Q}.
\]
Also, $\normtwo{P-Q}\leq
\epsilon_\rho+\normtwo{\rho-\sigma}+\epsilon_\sigma$.
Combining the two inequalities proves the claim.
\end{proof}

\begin{lemma}[\cref{lem:purity-to-unitary-choi}]
Every channel $\mathcal E$ from $A$ to itself satisfies
\[
\Delta(\mathcal E)
\leq e(\mathcal E)
\leq 5\sqrt{\Delta(\mathcal E)}.
\]
\end{lemma}
\begin{proof}[Proof of \cref{lem:purity-to-unitary-choi}]
Let $\rho=J(\mathcal E)$.
Choose a unitary channel $\mathcal V$ attaining the infimum in the definition of $e(\mathcal E)$.
Since $J(\mathcal V)$ is pure, the measurement
$\{J(\mathcal V),\Id-J(\mathcal V)\}$ accepts $J(\mathcal V)$
with probability one. Hence
$\Tr(J(\mathcal V)\rho)\geq
1-\normone{\rho-J(\mathcal V)}/2
=1-e(\mathcal E)/2$.
Since $\Tr(\rho^2)\geq\Tr(J(\mathcal V)\rho)^2$,
we have $\Delta(\mathcal E)\leq e(\mathcal E)$.

Conversely, let $p$ be the largest eigenvalue of $\rho$ and
$\ket\gamma$ a corresponding eigenvector.
Since $p\geq\Tr(\rho^2)=1-\Delta(\mathcal E)$,
$\normone{\rho-\proj\gamma}=2(1-p)\leq2\Delta(\mathcal E)$.
As $\Tr_A\rho=\Id_R/d_A$ and partial trace cannot increase
trace distance,
\[
\normone{\Tr_A\proj\gamma-\Id_R/d_A}
\leq\normone{\proj\gamma-\rho}\leq2\Delta(\mathcal E).
\]
Write $\ket\gamma$ in Schmidt form and choose a unitary $V$ such that
\[
\ket\gamma=\sum_{j=1}^{d_A}\sqrt{s_j}\ket{a_j}_A\ket{r_j}_R,
\qquad
(V\otimes\Id_R)\ket\Omega
=\frac1{\sqrt{d_A}}\sum_{j=1}^{d_A}\ket{a_j}_A\ket{r_j}_R.
\]
Then
\[
\norm{\ket\gamma-(V\otimes\Id_R)\ket\Omega}^2
=\sum_j\left(\sqrt{s_j}-\frac1{\sqrt{d_A}}\right)^2
\leq\sum_j\abs{s_j-1/d_A}\leq2\Delta(\mathcal E).
\]
The pure-state distance bound and the triangle inequality give
$e(\mathcal E)\leq2\Delta(\mathcal E)+2\sqrt{2\Delta(\mathcal E)}
\leq5\sqrt{\Delta(\mathcal E)}$,
where the last inequality uses $0\leq\Delta(\mathcal E)\leq1$.
\end{proof}

We allow the measurements to make $\QPSPACE$ oracle queries.
The following corollary is implicit in the previous works.

\begin{corollary}[QPSPACE implementation]
\label{cor:qpspace-quantum-or}
In \cref{lem:quantum-or}, suppose that $M_i$ is described by unitary QPT circuits with $\QPSPACE$ access that are generated uniformly from $i$ and a common polynomial-length
classical input. 
Suppose that $\xi$ consists of
polynomially many qubits and $\log N=\poly(\lambda)$.
Then the quantum OR test has a uniform QPT implementation
with access to $\QPSPACE$.
\end{corollary}
\begin{proof}
The unitary OR construction of \cite[Appendix~A]{EC:CheColSat25} can be executed using $\QPSPACE$.
The proof of \cite[Lemma~5.2]{EC:BMMMY25} extends this implementation to binary measurements given uniformly by polynomial-size circuits with $\QPSPACE$ access.
Since the family $\{M_i\}_{i=1}^N$ has such implementations, $\xi$ has polynomially many qubits, and $\log N=\poly(\lambda)$, this result gives the claimed
QPT implementation with access to $\QPSPACE$.
\end{proof}

\subsection{Proofs for CHFS derivatives}\label{app:chfsproofs}

\begin{lemma}[\cref{lem:two-sided-chfs-sum}]
Let $P,Q$ be orthogonal projections of rank
at most $r$ on the CHFS query register together with an arbitrary ancillary register. Then
\[
    \sum_{x\in\{0,1\}^d}
    \sum_{\delta}
    \normtwo{
        Q\bigl(
            \partial_{x,\delta}S_d^\Phi
            \otimes\Id
        \bigr)P
    }^2
    \leq
    9r,
\]
where, for each $x$, $\delta$ ranges over the
basis $\{i\phi_x,
    e_{x,1},ie_{x,1},
    \ldots,
    e_{x,2^d-1},ie_{x,2^d-1}\}$.
\end{lemma}

\begin{proof}[Proof of \cref{lem:two-sided-chfs-sum}]
We suppress the subscript $x$ in this proof; the distinct $x$'s are only
considered when we use the summation over $x\in\bit^d$.
By \cref{lem:explicit-chfs-derivatives} and the parallelogram identity, we have
\begin{align}
    &\normtwo{
        Q\bigl(
            \partial_{e_j}S_d^\Phi\otimes\Id
        \bigr)P
    }^2
    +
    \normtwo{
        Q\bigl(
            \partial_{ie_j}S_d^\Phi\otimes\Id
        \bigr)P
    }^2
    \nonumber\\
    &\qquad=
    4\normtwo{
        Q\bigl(
            \ketbra{u}{t_j}\otimes\Id
        \bigr)P
    }^2
    +
    4\normtwo{
        Q\bigl(
            \ketbra{t_j}{u}\otimes\Id
        \bigr)P
    }^2.
    \label{eq:real-imaginary-parallelogram}
\end{align}
Since $0\leq Q\leq\Id$, we have
\begin{align*}
    \sum_{j=1}^{2^d-1}
    \normtwo{
        Q\bigl(
            \ketbra{u}{t_j}\otimes\Id
        \bigr)P
    }^2
    &=
    \sum_{j=1}^{2^d-1}
    \Tr\left(
        P
        \bigl(\ketbra{t_j}{u}\otimes\Id\bigr)
        Q
        \bigl(\ketbra{u}{t_j}\otimes\Id\bigr)
        P
    \right)
    \\
    \leq&
    \sum_{j=1}^{2^d-1}
    \Tr\left(
        P\bigl(\proj{t_j}\otimes\Id\bigr)P
    \right)
    =
    \Tr\left(
        P(\Pi_T\otimes\Id)P
    \right).
\end{align*}
Similarly, by invariance of the Hilbert--Schmidt norm under adjoint and
the same calculation with $P$ and $Q$ interchanged,
\[
    \sum_{j=1}^{2^d-1}
    \normtwo{
        Q\bigl(
            \ketbra{t_j}{u}\otimes\Id
        \bigr)P
    }^2
    =
    \sum_{j=1}^{2^d-1}
    \normtwo{
        P\bigl(
            \ketbra{u}{t_j}\otimes\Id
        \bigr)Q
    }^2
    \leq
    \Tr\left(
        Q(\Pi_T\otimes\Id)Q
    \right).
\]
Therefore, summing \cref{eq:real-imaginary-parallelogram} over $j$ gives, for each $x$,
\begin{align*}
    &\sum_{j=1}^{2^d-1}
    \left(
        \normtwo{
            Q\bigl(
                \partial_{e_j}S_d^\Phi\otimes\Id
            \bigr)P
        }^2
        +
        \normtwo{
            Q\bigl(
                \partial_{ie_j}S_d^\Phi\otimes\Id
            \bigr)P
        }^2
    \right)
    \\
    &\qquad\leq
    4\Tr\left(
        P(\Pi_T\otimes\Id)P
    \right)
    +
    4\Tr\left(
        Q(\Pi_T\otimes\Id)Q
    \right).
\end{align*}

For the phase direction,
\cref{lem:explicit-chfs-derivatives} gives $\bigl(
        \partial_{i\phi}S_d^\Phi
    \bigr)^\dagger
    \bigl(
        \partial_{i\phi}S_d^\Phi
    \bigr)
    =
    \Pi_E$ and
hence,
\[
    \normtwo{
        Q\bigl(
            \partial_{i\phi}S_d^\Phi\otimes\Id
        \bigr)P
    }^2
    \leq
    \Tr\left(
        P(\Pi_E\otimes\Id)P
    \right).
\]

Restoring the subscript $x$ and summing over $x$, the mutual
orthogonality of the spaces $E_x\oplus T_x$ gives
\[
\sum_{x\in\{0,1\}^d}
    \sum_{\delta}
    \normtwo{
        Q\bigl(
            \partial_{x,\delta}S_d^\Phi\otimes\Id
        \bigr)P
    }^2
    \leq
    5\Tr P+4\Tr Q
    \leq
    9r,
\]
where we used $\sum_x\Pi_{E_x}\leq\Id$,
$\sum_x\Pi_{T_x}\leq\Id$, and
$\Tr(R\Pi R)\leq\Tr R$ for orthogonal projections $R$ and $\Pi$. The final inequality is due to the rank condition.
\end{proof}

\begin{lemma}[\cref{lem:one-sided-chfs-sum}]
Let $I$ be any register. For every vector $\xi$ on the CHFS query
register together with $I$,
\[
    \sum_{x\in\{0,1\}^d}
    \sum_{\delta}
    \norm{
        \bigl(
            \partial_{x,\delta}S_d^\Phi
            \otimes\Id_I
        \bigr)\xi
    }^2
    \leq
    5\cdot 2^d\norm\xi^2,
\]
where, for each $x$, $\delta$ ranges over the
basis $\{i\phi_x,
    e_{x,1},ie_{x,1},
    \ldots,
    e_{x,2^d-1},ie_{x,2^d-1}\}$.
\end{lemma}

\begin{proof}[Proof of \cref{lem:one-sided-chfs-sum}]
We suppress the subscript $x$ in this proof; the distinct $x$'s are only
considered when we use the summation over $x\in\bit^d$.
By \cref{lem:explicit-chfs-derivatives}, we have
\[
    \bigl(
        \partial_{e_j}S_d^\Phi
    \bigr)^\dagger
    \bigl(
        \partial_{e_j}S_d^\Phi
    \bigr)
    +
    \bigl(
        \partial_{ie_j}S_d^\Phi
    \bigr)^\dagger
    \bigl(
        \partial_{ie_j}S_d^\Phi
    \bigr)
    =
    4\left(
        \proj u+\proj{t_j}
    \right).
\]
Also, for the phase direction, we have $
    \bigl(
        \partial_{i\phi}S_d^\Phi
    \bigr)^\dagger
    \bigl(
        \partial_{i\phi}S_d^\Phi
    \bigr)
    =
    \Pi_E.$
Therefore, summing over the tangent basis gives, for each $x$,
\begin{align*}
    \sum_{\delta}
    \bigl(
        \partial_{\delta}S_d^\Phi
    \bigr)^\dagger
    \bigl(
        \partial_{\delta}S_d^\Phi
    \bigr)
    &=
    4(2^d-1)\proj u
    +
    4\Pi_T
    +
    \Pi_E
    \\
    &\leq
    (4\cdot 2^d-3)\Pi_E
    +
    4\Pi_T,
\end{align*}
where we used $\proj u\leq\Pi_E$.

Summing over all $x$, the mutual
orthogonality of the spaces $E_x\oplus T_x$ gives
\[
    \sum_{x\in\{0,1\}^d}
    \sum_{\delta}
    \bigl(
        \partial_{x,\delta}S_d^\Phi
    \bigr)^\dagger
    \bigl(
        \partial_{x,\delta}S_d^\Phi
    \bigr)
    \leq
    5\cdot 2^d\Id.
\]
Hence, expanding the left hand side becomes
\[
\sum_{x\in\{0,1\}^d}
    \sum_{\delta}
    \norm{
        \bigl(
            \partial_{x,\delta}S_d^\Phi
            \otimes\Id_I
        \bigr)\xi
    }^2
    =
    \left\langle
        \xi,
        \left(
            \sum_{x\in\{0,1\}^d}
            \sum_{\delta}
            \bigl(
                \partial_{x,\delta}S_d^\Phi
            \bigr)^\dagger
            \bigl(
                \partial_{x,\delta}S_d^\Phi
            \bigr)
            \otimes\Id_I
        \right)
        \xi
    \right\rangle
\]
which is bounded above by $5\cdot 2^d\norm\xi^2$, proving the lemma.
\end{proof}

\begin{lemma}[\cref{lem:gradient-energy-bound}]
For every CHFS state family $\Phi$,
\[
 \|\nabla F_d^{\mathcal A}(\Phi)\|_{\mathrm{HS}}^2=
\sum_{x\in\{0,1\}^d}
\sum_{\delta}
\normtwo{
\partial_{x,\delta}\rho_\Phi
}^2
\leq
72q_d^2
\bigl(
1+2^d\eta(\Phi)
\bigr).
\]
\end{lemma}

\begin{proof}[Proof of \cref{lem:gradient-energy-bound}]
By the preceding derivative identity and Cauchy-Schwarz,
\begin{align}
\normtwo{
\partial_{x,\delta}\rho_\Phi
}^2
\leq
q_d\sum_{t=1}^{q_d}
\normtwo{
\Tr_B\left(
\ketbra{\Delta_{t,x,\delta}}{\Xi_\Phi}
+
\ketbra{\Xi_\Phi}{\Delta_{t,x,\delta}}
\right)
}^2.
\label{eq:query-cauchy-schwarz}
\end{align}
Recall $\ket{\Xi_\Phi}
    =
    \ket{\Xi_\Phi^{\mathrm{id}}}
    +
    \ket{e_\Phi}$. To deal with them separately, define
\[
X_{t,x,\delta}
\defeq
\Tr_B\left(
\ketbra{\Delta_{t,x,\delta}}
{\Xi_\Phi^{\mathrm{id}}}
\right),
Y_{t,x,\delta}
\defeq
\Tr_B\left(
\ketbra{\Delta_{t,x,\delta}}{e_\Phi}
\right).
\]
Then, we have $\Tr_B\left(
\ketbra{\Delta_{t,x,\delta}}{\Xi_\Phi}
+
\ketbra{\Xi_\Phi}{\Delta_{t,x,\delta}}
\right)
=
X_{t,x,\delta}
+
X_{t,x,\delta}^\dagger
+
Y_{t,x,\delta}
+
Y_{t,x,\delta}^\dagger,$
and therefore
\begin{align}
\normtwo{
\Tr_B\left(
\ketbra{\Delta_{t,x,\delta}}{\Xi_\Phi}
+
\ketbra{\Xi_\Phi}{\Delta_{t,x,\delta}}
\right)
}^2
\leq
8\normtwo{X_{t,x,\delta}}^2
+
8\normtwo{Y_{t,x,\delta}}^2.
\label{eq:z-product-error-bound}
\end{align}

We first bound the terms $X_{t,x,\delta}$.
Let $\{\ket a\}_{a=1}^{d_A}$ be the orthonormal basis of $A$ used to
define $\ket\Omega_{AR}$. Define
\[
    \ket{s_a}
    \defeq
    W_{>t,\Phi}
    \bigl(\partial_{x,\delta}S^\Phi\bigr)
    W_{<t,\Phi}
    \bigl(\ket a_A\ket0_B\bigr).
\]
Then
\[
    \ket{\Delta_{t,x,\delta}}
    =
    \frac1{\sqrt{d_A}}
    \sum_{a=1}^{d_A}
    \ket{s_a}\ket a_R,\quad
    \ket{\Xi_\Phi^{\mathrm{id}}}
    =
    \frac1{\sqrt{d_A}}
    \sum_{b=1}^{d_A}
    V_\Phi\ket b_A\ket{\chi_\Phi}_B\ket b_R.
\]
Therefore,
\[
    X_{t,x,\delta}
    =
    \frac1{d_A}
    \sum_{a,b=1}^{d_A}
    (\Id_A\otimes\bra{\chi_\Phi})\ket{s_a}
    \bra{ b}V_\Phi^\dagger
    \otimes\ketbra {a} {b}_R.
\]
Since the elements $\{\ketbra {a} {b}_R\}_{a,b}$ can be regarded as orthonormal matrices with respect to the
Hilbert-Schmidt inner product,
\[
    \normtwo{X_{t,x,\delta}}^2
    =
    \frac1{d_A^2}
    \sum_{a,b=1}^{d_A}
    \norm{
        (\Id_A\otimes\bra{\chi_\Phi})\ket{s_a}
    }^2.
\]

Define the projections regarding the initial states and the final states and the operators $W_{<t}$ and $W_{>t}$ by
\begin{align*}
P_{t,\Phi}
&\defeq
W_{<t,\Phi}
\bigl(\Id_A\otimes\proj{0}_B\bigr)
W_{<t,\Phi}^\dagger,
\\
Q_{t,\Phi}
&\defeq
W_{>t,\Phi}^\dagger
\bigl(
\Id_A\otimes
\ketbra{\chi_\Phi}{\chi_\Phi}_B
\bigr)
W_{>t,\Phi}.
\end{align*}
Both $P_{t,\Phi}$ and $Q_{t,\Phi}$ are orthogonal projections of rank
$d_A$.
Then we can write
\begin{align*}
\normtwo{X_{t,x,\delta}}^2
&=
\frac1{d_A^2}
\sum_{a,b}
\norm{
(\Id_A\otimes\bra{\chi_\Phi})\ket{s_a}
}^2
\\
&=
\frac1{d_A}
\sum_a
\norm{
(\Id_A\otimes\bra{\chi_\Phi})
W_{>t,\Phi}
\bigl(\partial_{x,\delta}S^\Phi\bigr)
W_{<t,\Phi}
\bigl(\ket a_A\ket0_B\bigr)
}^2
\\
&=
\frac1{d_A}
\Tr\left(
P_{t,\Phi}
\bigl(\partial_{x,\delta}S^\Phi\bigr)^\dagger
Q_{t,\Phi}
\bigl(\partial_{x,\delta}S^\Phi\bigr)
P_{t,\Phi}
\right)
\\
&=
\frac1{d_A}
\normtwo{
Q_{t,\Phi}
\bigl(\partial_{x,\delta}S^\Phi\bigr)
P_{t,\Phi}
}^2.
\end{align*}
Since $P_{t,\Phi}$ and $Q_{t,\Phi}$ have rank $d_A$,
\cref{lem:two-sided-chfs-sum} gives
\begin{align}
\sum_{x,\delta}
\normtwo{X_{t,x,\delta}}^2
=
\frac1{d_A}
\sum_{x,\delta}
\normtwo{
Q_{t,\Phi}
\bigl(\partial_{x,\delta}S^\Phi\bigr)
P_{t,\Phi}
}^2\leq
9.
\label{eq:product-derivative-total}
\end{align}

For the error term,
$\norm{e_\Phi}\leq\sqrt{\eta(\Phi)}$ gives
\[
\normtwo{Y_{t,x,\delta}}^2
=\normtwo{\Tr_B\left(
\ketbra{\Delta_{t,x,\delta}}{e_\Phi}
\right)}^2
\leq \norm{\Delta_{t,x,\delta}}^2 \norm{e_\Phi}^2
\leq
\eta(\Phi)
\norm{\Delta_{t,x,\delta}}^2
\]
where we use $\normtwo{\Tr_B(\ketbra{x}{y})} \le \norm{x}\norm{y}.$
The pre-query state
\[
(W_{<t,\Phi}\otimes\Id_R)
\bigl(\ket\Omega_{AR}\ket0_B\bigr)
\]
is a unit vector. Including the reference register and all other
non-query registers in the ancillary register of
\cref{lem:one-sided-chfs-sum}, and using invariance of the norm under
$W_{>t,\Phi}$, gives
\begin{align}
\sum_{x,\delta}
\normtwo{Y_{t,x,\delta}}^2
\leq
5\cdot 2^d\eta(\Phi).
\label{eq:error-derivative-total}
\end{align}

Summing \cref{eq:query-cauchy-schwarz} over $x,\delta$, and using
\cref{eq:z-product-error-bound,eq:product-derivative-total,eq:error-derivative-total},
gives
\[
8q_d^2
\bigl(
9+5\cdot 2^d\eta(\Phi)
\bigr)
\leq
72q_d^2
\bigl(
1+2^d\eta(\Phi)
\bigr),
\]
which proves the claim.
\end{proof}

\subsection{Proofs for the distinguisher}\label{app:proofsdisting}

The following lemma ensures the accuracy of the above circuit.

\begin{lemma}[\cref{lem:average-choi-preparation}]
Given $\lambda,k$ and the descriptions
$\widetilde S_{\leq\tau}$, the above procedure prepares a purification of $\widehat\mu_k$.
Except with probability at most $(\tau+1)2^{-2\lambda}$ over
tomography, for every key $k$,
$\normone{\widehat\mu_k-\mu_k}\leq1/200$.
\end{lemma}

\begin{proof}[Proof of \cref{lem:average-choi-preparation}]
The preparation claim follows from the construction.
Fix the tomography descriptions, which are used for all keys. Apply
\cref{lem:chfs-design-simulation} to
$\widetilde{\mathcal G}_k$ acting on one half of $\ket\Omega$ for the CHFS oracles with input length $d>\tau$.
The computation makes at most $q$ CHFS oracle queries, each of
length at most $\lceil\lambda^a\rceil$. Thus
\begin{align}
\label{eqn:deltacondition}
\normone{
\widehat\mu_k-
\EE_{\Phi'_{>\tau}}J(\widetilde{\mathcal G}_k^{\Phi'_{>\tau}})
}
\leq KqC^{q(\lceil\lambda^a\rceil+1)}\delta
\leq\frac1{400},
\end{align}
where $K,C$ are the universal constants in that lemma.
Since $q\leq\lambda^a$, this choice requires only
$\log(1/\delta)=O(\lambda^{2a})$.

Now suppose that all tomography estimates have error at most $1/400q$. For every key $k$ and every $\Phi'_{>\tau}$, a hybrid
argument over the queries of size $\le \tau$ gives
\[
\normone{
J(\widetilde{\mathcal G}_k^{\Phi'_{>\tau}})
-\rho_{k,(\Phi_{\leq\tau},\Phi'_{>\tau})}
}
\leq\sum_{d\leq\tau}\frac{q_d}{400q}\leq\frac1{400}.
\]
Averaging over $\Phi'_{>\tau}$ and combining
the two bounds gives
$\normone{\widehat\mu_k-\mu_k}\leq1/200$.
A union bound over all the process tomography procedures gives failure probability
at most $(\tau+1)2^{-2\lambda}$.
\end{proof}

\begin{lemma}[\cref{lem:QPTimplement}]
The distinguisher is a uniform QPT algorithm with oracle access to
$\mathcal V$ and $\mathcal O^\Phi=(S^\Phi,\QPSPACE)$.
\end{lemma}
\begin{proof}[Proof of \cref{lem:QPTimplement}]
The tomography procedure runs in polynomial time as shown in \cref{subsec:simulating-g}.
For fixed tomography descriptions, the preparation for $\widehat\mu_k$ and the $r=O(\lambda)$ swap tests give unitary QPT implementations of the measurements $M_k$ with access to $\QPSPACE$, whose circuits can be generated uniformly from $k$ and $\widetilde S_{\leq\tau}$.
There are $2^\lambda$ measurements, so \cref{cor:qpspace-quantum-or} gives a QPT implementation
of the quantum OR test with access to $\QPSPACE$.
All remaining operations have polynomial size.
\end{proof}

\begin{lemma}[\cref{lem:haar-case}]
If $\mathcal V$ is a Haar-random unitary, the distinguisher
outputs $1$ with negligible probability.
\end{lemma}

\begin{proof}[Proof of \cref{lem:haar-case}]
Let $\mathcal V=\mathcal U$ for a Haar-random unitary $U$.
The purity test always passes because $J(\mathcal U)$ is pure.
It therefore suffices to bound the acceptance probability of
the quantum OR test.

Fix the original oracle and the tomography outcomes.
Every $\widehat\mu_k$ is then fixed independently of $U$.
For each key $k$, a swap test on $J(\mathcal U)$ and
$\widehat\mu_k$ passes with probability
$(1+\Tr(J(\mathcal U)\widehat\mu_k))/2$.
Haar invariance gives
\[
\EE_U J(\mathcal U)=\frac{\Id_{AR}}{2^{2n}},
\qquad
\EE_U\left[\frac{1+\Tr(J(\mathcal U)\widehat\mu_k)}2\right]
=\frac12+\frac1{2^{2n+1}}.
\]
This swap test uses one query to $U$. Hence
\cref{cor:haar-concentration} gives, for each $k$,
\[
\Prb_U\left[
\frac{1+\Tr(J(\mathcal U)\widehat\mu_k)}2>\frac7{12}
\right]
\leq\exp(-\Omega(2^n)).
\]
A union bound over the keys shows that, except with probability
at most $2^\lambda\exp(-\Omega(2^n))=\negl(\lambda)$,
the swap-test acceptance probability is at most $7/12$
for every key. We only consider this case. Then, by Chernoff inequality 
(\cref{lem:bernoulli-chernoff}) and $r=\lceil143\lambda\rceil$,
each $M_k$ accepts with probability at most $2^{-3\lambda}$.
Thus the quantum OR test accepts and the distinguisher outputs 1
with probability at most $4\cdot2^\lambda\cdot2^{-3\lambda}$ by \cref{lem:quantum-or}.
Overall, the distinguisher outputs $1$
with negligible probability.
\end{proof}

\begin{lemma}[\cref{lem:candidate-case}]
Fix an oracle $\Phi$ for which the conclusion of
\cref{lem:purity-concentration-dichotomy} holds. If
$\mathcal V=\mathcal G_{k^\star}^{\mathcal O^\Phi}$ for a uniformly
random $k^\star$, then the distinguisher outputs 1 with non-negligible probability. More precisely,
$\Prb[\mathcal D^{\mathcal V,\mathcal O^\Phi}=1]
\geq\frac12\lambda^{-(4a+9)}-\negl(\lambda)
$ for all sufficiently large $\lambda$.
\end{lemma}

\begin{proof}[Proof of \cref{lem:candidate-case}]
We first consider the mixedness case, i.e., it holds that $\EE_k[1-\Tr(\rho_{k,\Phi}^2)]\geq\lambda^{-(4a+9)}$. In this case,
the purity test alone outputs $1$ with probability at least
$\frac12\lambda^{-(4a+9)}$.

Otherwise, \cref{lem:purity-concentration-dichotomy} gives
$\EE_k\normone{\rho_{k,\Phi}-\mu_k}\leq\lambda^{-1}$.
By Markov's inequality, all but a
$100\lambda^{-(4a+9)}+100\lambda^{-1}$ fraction of keys satisfy the near-purity and closeness conditions simultaneously:
\[
1-\Tr(\rho_{k,\Phi}^2)\leq\frac1{100},
\qquad
\normone{\rho_{k,\Phi}-\mu_k}\leq\frac1{100}.
\]
Condition on successful tomography. 
By \cref{lem:average-choi-preparation}, every such key satisfies
$\normone{\rho_{k,\Phi}-\widehat\mu_k}\leq3/200$, and consequently
\[
\Tr(\rho_{k,\Phi}\widehat\mu_k)
\geq\Tr(\rho_{k,\Phi}^2)
-\normone{\rho_{k,\Phi}-\widehat\mu_k}
\geq\frac{39}{40}.
\]
For each such key $k$, every swap test in $M_k$ passes
with probability at least $79/80$.
Since $M_k$ rejects only if fewer than $2r/3$ of its $r$
independent swap tests pass, the Chernoff-Hoeffding bound
(\cref{lem:bernoulli-chernoff}) gives a rejection probability
of at most $2^{-3\lambda}$.
For the actual key $k^\star$ satisfying both the near-purity and closeness conditions,
the quantum OR test (\cref{lem:quantum-or}) thus accepts with probability at least
$(1-2^{-3\lambda})^2/7$.

The purity test acts on separate copies of $J(\mathcal V)$, and its
failure also causes the distinguisher to output $1$. Therefore,
averaging over the key and accounting for tomography failure gives
\[
\Prb[\mathcal D^{\mathcal V,\mathcal O^\Phi}=1]
\geq
\left(1-100\lambda^{-(4a+9)}-\frac{100}{\lambda}\right)
\frac{(1-2^{-3\lambda})^2}{7}
-(\tau+1)2^{-2\lambda}
\]
where $(\tau+1)2^{-2\lambda}$ is due to the tomography failure probability from \cref{lem:average-choi-preparation}.
This is bounded below by a positive constant for all sufficiently
large $\lambda$, proving the claim.
\end{proof}
\fi

\end{document}